\documentclass[11pt]{article}

\usepackage[margin=1in]{geometry}

\usepackage{authblk}
\usepackage{hyperref}
\usepackage{url}
\usepackage{natbib}

\usepackage[utf8]{inputenc} 
\usepackage[T1]{fontenc}    
\usepackage{hyperref}       
\usepackage{url}            
\usepackage{booktabs}       
\usepackage{amsfonts}       
\usepackage{nicefrac}       
\usepackage{microtype}     
\usepackage{xcolor}      
\usepackage{graphicx}
\usepackage{multirow}
\usepackage{amsmath,amsfonts,amsmath,amssymb,amsthm}
\usepackage{makecell}
\newtheorem{theorem}{Theorem}
\newtheorem{lemma}{Lemma}
\newtheorem{proposition}{Proposition}
\newtheorem{remark}{Remark}
\newtheorem{corollary}{Corollary}
\newtheorem{definition}{Definition}

\usepackage{algorithm}
\usepackage{algpseudocode}
\usepackage{xspace}

\def\IsDraft{} 

\ifdefined\IsDraft
    \newcommand{\authnote}[2]{{\bf [{\color{red} #1's Note:} {\color{blue} #2}]}}
\else
    \newcommand{\authnote}[2]{}
\fi

\newcommand{\eps}{\varepsilon}
\newcommand{\calF}{\mathcal{F}}
\newcommand{\calX}{\mathcal{X}}

\newcommand{\arcmark}{\textsc{ArcMark}\xspace}
\newcommand{\simplemark}{\textsc{SimpleMark}\xspace}

\title{\textsc{SimpleMark}: Fast Multi-Bit Text Watermarking under $f$-Divergence Constraints}

\author[1]{Benjamin D. Kim}
\affil[1] {Massachusetts Institute of Technology}

\author[2]{Wanrong Zhang}
\author[3]{Weitong Ruan}
\affil[2,3] {Amazon}

\author[4]{Lav R. Varshney}
\affil[4] {Stony Brook University}

\author[5]{Daniel Alabi}
\affil[5] {University of Illinois Urbana-Champaign}

\date{}

\begin{document}

\maketitle

\begin{abstract}

We introduce a framework for multi-bit text watermarking with security defined directly through $f$-divergence from the base language model distribution. Unlike prior approaches that focus on average-key distortion-freeness or a particular statistical distance, our formulation supports general $f$-divergences, including total variation and KL divergence, and enforces the guarantee for each realized key and embedded message. We develop a coding-based watermarking scheme that optimally biases next-token distributions subject to a prescribed divergence budget, and characterize the resulting tradeoff between embedding rate, decoding reliability, and statistical security. Experimentally, we compare our method against prior multi-bit watermarking schemes across modern language models and payload regimes. Our approach achieves substantially lower watermark detectability while maintaining competitive message-recovery performance and generation quality. Our results provide a unified view of secure multi-bit watermarking and recover several commonly used security notions as special cases.

\end{abstract}

\section{Introduction}
\label{sec:intro}

Large language models (LLMs) can now generate fluent, high-quality text at scale.
 This creates new challenges for provenance, accountability, and policy enforcement~\citep{Longpre24}.
Text watermarking offers a promising response: a model provider can embed a hidden signal into generated
text so that the output can later be detected or decoded~\citep{KirchenbauerGWK23}.
While early watermarking schemes primarily focus on \textit{zero-bit}
detection (to ascertain the presence of a watermark)~\citep{ChristGZ23,kuditipudi2024robust,KirchenbauerGWK23}, many
applications require \textit{multi-bit} watermarking, where the generated text carries a message such as a user
identifier, model identifier, timestamp, or policy tag~\citep{CMBFK07}.

Multi-bit watermarking introduces a fundamental tension: increasing the payload generally requires a stronger
perturbation of the model's next-token distribution, making the watermark both easier to recover and easier to detect by an adversary~\citep{pmlr-v267-jiang25j}. 
Conversely, imposing stronger distributional constraints on the watermarked model, relative to the unwatermarked model,
makes the watermark less detectable but limits the amount of information
that can be reliably embedded. This raises a basic question: \textit{what are the capacity and detectability of a multi-bit text watermark while being as close as possible to the unwatermarked distribution?} 
We study this question by comparing how far from the unwatermarked distributions different watermarks perturb down to individual token distributions and fundamentally how much information one can embed. 
To measure distance from the base model, we adopt
per-key rather than average-key notions \citep{gilani2026arcmarkmultibitllmwatermark}. This is because
in adversarial settings~\citep{chang-etal-2025-watermark,zhang2024watermarks,AKV03}, security should hold for the actual distribution produced in a realized execution, not after averaging over secret keys~\citep{KL14}. An average-key guarantee can mask large key-dependent deviations that give a distinguisher substantial advantage, motivating a stronger per-key notion of statistical indistinguishability~\citep{CMBFK07}.

Motivated by this observation, we introduce \simplemark, a multi-bit text watermarking scheme that enforces
an $f$-divergence $\varepsilon$-budget at each generation step. 
\simplemark uses low density parity check (LDPC) codes with a 
geometry-aware (simplex) optimal transport embedding rule and maximum-likelihood decoding. The simplex geometry allows the reduction of the per-token embedding step to a two-scalar tilt computable in $O(1)$, while the LDPC coding allows for improved, coding-theoretic reliability in recovery. 
Figure~\ref{fig:diagram} graphically illustrates \simplemark.
In Section~\ref{sec:main-algorithms}, we discuss the components of \simplemark in more detail.

\textbf{Our contributions:}
First, we formalize multi-bit text watermarking as communication through a constrained watermark channel and
identify the mutual information governing reliable message recovery. Second, we show a channel-collapse result:
exact per-key distortion-freeness forces zero embedding capacity (i.e., no scheme can embed even a single bit)
so distortion is not a defect of existing multi-bit watermarks but the
price of embedding, and the right goal is to control it explicitly, through $f$-divergence constraints yielding a tunable capacity--detectability tradeoff. Third, we demonstrate that average-key distortion-freeness and other such metrics do not imply per-key undetectability and can 
lead to highly detectable realized-key behavior. Finally, we evaluate \simplemark~on modern language models showing it achieves competitive or state-of-the-art performance on reasoning and advanced tasks, recovery, perplexity, and more while substantially reducing the distributional
deviation between the
watermarked model and the base model.

\begin{figure}
    \centering
    \includegraphics[width=0.9\linewidth]{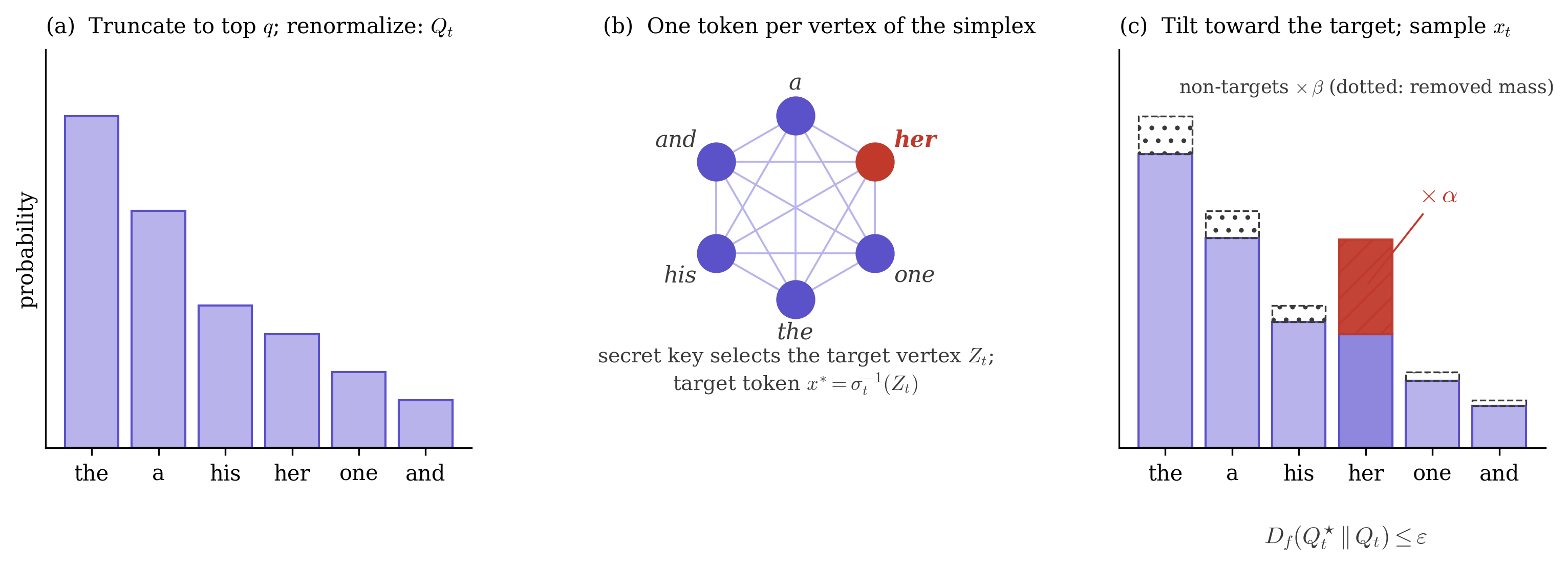}
    \caption{Watermark encoding and embedding process for \simplemark: the watermarker starts by truncating the language model support according to the top-$q$, and places all the tokens on a $q-$ary simplex. Using the message $M$ and the generator matrix $G$ to create a codeword $C$, and computes the target vertex on the simplex, depicted in red in the diagram. The watermarker biases the watermarked distribution towards this token by $\alpha$, and reduces the probability of sampling other tokens by $\beta$ while remaining in the $\varepsilon$ parameter $f$-divergence budget, $D_f (Q_t^* \Vert Q_t) \le \varepsilon$, and samples from the watermarked distribution $Q_t^*$. The message can then be decoded with one of the decoders. A detailed procedure is in Section~\ref{sec:main-algorithms}.}
    \label{fig:diagram}
\end{figure}

\subsection{Related Work}
\label{sec:relatedwork}

\paragraph{Multi-bit Watermarks} 
Recently, a variety of multi-bit watermarks have been developed for language models~\citep{yoo-etal-2024-advancing, pmlr-v267-jiang25j,qu2025provably,xu2025robust, jiang2026mirrormark, cui2026mc}. Two prominent ones are
BiMark \citep{feng2025bimark} and ArcMark \citep{gilani2026arcmarkmultibitllmwatermark}. A pivotal advancement is the development of provably robust frameworks; \citep{qu2025provably} use error-correction codes (ECC) to provide formal guarantees within specific edit-distance bounds, whereas \citep{cohen2025watermarking} establishes resistance against adaptive prompting. Recent work such as \citep{xu2026xmark} further optimizes this tradeoff by employing block-wise encoding to ensure high recovery rates even in short-form textual outputs, while \citep{guo2026feedback} incorporates feedback coding. To the best of our knowledge, these previous works
do not provide an explicit per-key mechanism for directly tuning a prescribed divergence budget (e.g., for total variation or KL) at each generation step.

\begin{table}[t]
\centering
\footnotesize
\setlength{\tabcolsep}{3.5pt}
\renewcommand{\arraystretch}{1.15}
\begin{tabular}{@{}lcc@{}}
\toprule
Watermark & \makecell{Per-token watermark-\\specific overhead} & \makecell{Decoding cost} \\
\midrule
ArcMark      & $O(I\,|\mathcal{V}|\,r)$
             & $O\!\left(n(|\mathcal{V}|+2^{b})\right)$ \\
BiMark       & $O(d\,|\mathcal{V}|)$
             & $O(nd+b)$ \\
StealthInk   & $O(|\mathcal{V}|)$
             & $O\!\left(n|\mathcal{V}|+(b/c)2^{c}\right)$ \\
MC$^2$Mark   & $O\!\left(m|\mathcal{V}|+m\tbinom{n'}{n'/2}n'\right)$
             & $O(n\,m\,n'+b)$ \\
QuantileMark & $O(|\mathcal{V}|\log|\mathcal{V}|+M)$
             & $O\!\left(n(|\mathcal{V}|\log|\mathcal{V}|+M)+(b/m')M\right)^{\dagger}$ \\
MirrorMark   & $O(L\,|\mathcal{V}|)$
             & $O\!\left(n(b/c+2^{c}L)\right)$ \\
\midrule
\textsc{SimpleMark} & $O_{q}(1)$
             & $O\!\left(n|\mathcal{V}|+2^{b}n\right)^{\dagger}$ \\
\bottomrule
\end{tabular}
\caption{Complexity comparison between \simplemark and other known multi-bit watermark schemes. $^{\dagger}$Note the decoder requires a model replay. We assume $q = 53$ as a small constant. With respect to the complexities, the symbols can be defined as $|\mathcal{V}|$: vocabulary; $n$: text length; $b$: payload bits. Method-specific symbols: $r$: side-information resolution and $I$: Sinkhorn iteration count
for ArcMark \citep{gilani2026arcmarkmultibitllmwatermark}; $d$: reweighting
layers for BiMark \citep{feng2025bimark} ($d{=}10$); $m$: channel layers and
$n'$: segment length in bits for MC$^2$Mark \citep{cui2026mc} ($m{=}10$;
$n'$ unspecified in their experiments); $c$: bits per message chunk for
StealthInk \citep{pmlr-v267-jiang25j} and MirrorMark
\citep{jiang2026mirrormark} ($c{=}1$ and $c{=}2$ respectively; both papers
denote this $m$); $L$: tournament layers for MirrorMark ($L{=}30$). $M$: quantile bins per step for QuantileMark \citep{zhu2026quantilemark} ($M = 2^{m'}$ with $m'$ bits
per message symbol).}
\label{tab:complexity}
\end{table}

\paragraph{Watermarking Capacity}

Our work connects to the information-theoretic literature of watermarking, data hiding,
and steganography, where watermarking is modeled as communication over a constrained channel:
the encoder embeds a message into a host signal subject to a distortion or detectability constraint,
and the decoder attempts to recover it~\citep{CACHIN200441,CMBFK07}. Early work characterized achievable rates for
information hiding under distortion and attack constraints~\citep{MoulinO03,Moulin03,Moulin05,MoulinK2005},
and coding-based schemes such as quantization-index modulation showed how structured codes can
approach these limits~\citep{CW06}.

In the language model setting, the host distribution is the base model's next-token distribution, so the
relevant constraint is distributional rather than perceptual.  Recent work has begun to analyze LLM watermarking through this lens~\citep{gilani2026arcmarkmultibitllmwatermark, he2025theoretically}.  We further extend this capacity viewpoint
of multi-bit text watermarking by comparing average-key
distortion-freeness, uniform per-key total-variation closeness,
and perfect distortion-freeness.  Our results show that
stronger per-key distortion-freeness constraints reduce the mutual information of the induced watermark
channel, and exact per-key distortion-freeness forces zero embedding capacity.

\section{Preliminaries and Notation}
\label{sec:prelims}

First, we introduce some useful notation and background, many of which can be found in standard textbooks~\citep{ross98, gut2013probability, MitzenmacherUp05,keener2010theoretical}.

\paragraph{$f$-divergences.}
Let $f : (0, \infty) \to \mathbb R$ be convex with $f(1) = 0$. For
distributions $P, Q$ on $\calX$ with $P \ll Q$, the $f$-divergence
of $P$ from $Q$ is
$D_f(P \Vert Q) := \sum_{x \in \calX} Q(x)\, f\big(\tfrac{P(x)}{Q(x)}\big)$.
Every $f$-divergence satisfies the data-processing inequality:
$D_f(\kappa(P) \Vert \kappa(Q)) \le D_f(P \Vert Q)$ for any (possibly
randomized) map $\kappa$ \citep{polyanskiy2025information}. We next discuss the
two $f$-divergences this work focuses on.

\paragraph{Total variation distance.}
For probability measures $P,Q$ on measurable space $(\calX, \calF)$,
we define
$d_{\mathrm{TV}}(P,Q)
:=
\sup_{A \in \mathcal F} |P(A)-Q(A)|
=
\frac12 \sum_{x\in\mathcal X} |P(x)-Q(x)|.$
Equivalently, total variation distance characterizes the maximum distinguishing
advantage of any adversarial attacker trying to decide whether a sample was drawn from
$P$ or $Q$ \citep{KL14}. In particular, for any possibly randomized distinguisher
$\mathcal D : \mathcal X \to \{0,1\}$,
\[
\sup_{\mathcal D}
\left|
\Pr_{X\sim P}[\mathcal D(X)=1]
-
\Pr_{X\sim Q}[\mathcal D(X)=1]
\right|
=
d_{\mathrm{TV}}(P,Q).
\]
Relating to this work, we want the TV distance of the watermarked and unwatermarked distributions to be as small as possible, so an adversary has the smallest possible advantage in determining whether outputted text is watermarked or not.

\paragraph{Kullback--Leibler divergence.}
For distributions $P, Q$ on a finite alphabet $\calX$ with
$P \ll Q$, the Kullback--Leibler divergence of $P$ from $Q$ is
$D_{KL}(P \Vert Q) := \sum_{x \in \calX} P(x) \log \tfrac{P(x)}{Q(x)}$.
In our setting, $KL(P \Vert Q)$ is the expected excess log-loss
incurred when text generated by the watermarked distribution $P$
is scored by the base model $Q$, and thus relates directly to
perplexity: writing $H(P)$ for the entropy of $P$,
$\mathrm{PPL} \;=\; \exp\!\big(\mathbb E_{x \sim P}[-\log Q(x)]\big)
  \;=\; e^{H(P)}\, e^{KL(P \Vert Q)}.$
Here, $e^{KL(P \Vert Q)}$ is related to the
multiplicative perplexity inflation introduced by watermarking.

\paragraph{Watermarked distribution.}
Fix $\varepsilon \ge 0$. A watermarked distribution $Q_t^\star = P_{X_t\mid M=m,V_t=v}$ is said to be (uniformly per-key) \emph{$\varepsilon$-distortion-free} if for every time $t$, every message $m\in\mathcal{M}$, and every realized key value $v$,
$D_{f}
\bigl(
P_{X_t\mid M=m,V_t=v}
\,\Vert\,
Q_t
\bigr)
\le \varepsilon.$
Equivalently, when the encoder depends on $(m,v)$ only through the state $Z_t=z_t(m,v)$, the distribution is (uniformly per-key) \emph{$\varepsilon$-distortion-free} iff
$D_{f}
\bigl(
P_{X_t\mid Z_t=z}
\,\Vert\,
Q_t
\bigr)
\le \varepsilon
\qquad \text{for every realizable } z\in\mathcal{Z},\ \forall t.$

For our watermark, we can set $\eps$ to be a constant or to be cryptographically negligible (depending on the target security application). When $\eps = 0$, we obtain \emph{perfect} distortion-freeness.

\section{Main Results}
\label{sec:theory}

This section contains our main theoretical results. We begin by describing \simplemark's encoding, embedding, and decoding algorithms.

\subsection{Algorithms}
\label{sec:main-algorithms}

We describe \simplemark~being used in practice. Its components include:
\begin{itemize}
    \item \textbf{Encoder:} The message is coded with an LDPC code and then combined with the shared secret to compute the target vertex.
    \item \textbf{Embedding:} at each time $t$, the watermarked distribution $Q_t^\star$ is computed in closed form while adhering to the divergence budget. The watermarked token $X_t$ is sampled.
    \item \textbf{Decoding:} a decoder uses a likelihood decoder to compute $\hat M$. Alternatively, we provide a plurality decoder where only the received symbols are required.
\end{itemize}
Encoding and embedding are given in algorithm form in Appendix \ref{sec:encoder-decoder}.

\paragraph{Message Encoding} The watermarker begins with a $k$-symbol message $m$ and encodes it with a $q$-ary LDPC code, resulting in a length-$n$ codeword $C_1 \dots C_n$ over $\mathbb{F}_q$ for a prime $q$. For a $k \times n$ systematic generator matrix $G$ corresponding to a sparse parity-check matrix $H$, the resulting codeword is $C = m \cdot G$. The generator matrix can be considered public information.
At each step $t$, the shared secret $S_t = (V_t, \sigma_t)$ is derived by hashing the secret key with the $h$ previously generated tokens, so both parties recompute it from the text itself. It consists of a key value $V_t$ uniform on $\mathcal{R} = \{0, \dots, q-1\}$ and a keyed bijection $\sigma_t : \mathcal{S}_t \to \mathcal{R}$, where $\mathcal{S}_t \subset \mathcal{V}$ is the top-$q$ truncated support of the vocabulary $\mathcal{V}$, so $|\mathcal{S}_t| = q$ and each vertex receives exactly one token. We identify codeword symbols, key values, and colored tokens with the vertices of the $(q-1)$-simplex: the codeword symbol $C_t \in \mathbb{F}_q$ is mapped to vertex $C_t$, the key value $V_t$ to vertex $V_t$, and the token $x \in \mathcal{S}_t$ to vertex $\sigma_t(x)$.
The channel input vertex at time $t$ is then
\begin{equation}
    Z_t = \left( C_t + V_t \right) \bmod q.
    \label{eq:simplemark_channel_input}
\end{equation}
This vertex determines the target location on the simplex toward which the embedding step biases the next-token distribution, while the divergence constraint controls how far the resulting watermarked distribution may deviate from the base model distribution.

\paragraph{Message Embedding}
\simplemark~constructs a watermarked distribution $Q_t^\star$ by solving a constrained minimization problem which admits a closed-form solution.

For each token, let $\sigma_t(x)$ be the vertex assigned to $x$, where $\sigma_t$ is the random bijection shared between the watermarker and detector at time $t$, and let $Z_t$ be the target
vertex defined in \eqref{eq:simplemark_channel_input}. We define the embedding cost
$d_t(x) = \mathbf{1}\!\left[\sigma_t(x) \neq Z_t\right],$
which is the discrete metric on the vertex set, since all pairs of vertices are equidistant. The
watermarked distribution is then chosen as
\begin{equation}
    Q_t^\star
    \in
    \arg\min_{P \in \Delta(\mathcal S_t)}
    \sum_{x \in \mathcal S_t} P(x) d_t(x)
    \quad
    \text{s.t.}
    \quad
    D(P,Q_t) \le \varepsilon ,
    \label{eq:simplemark_embedding}
\end{equation}
where $\varepsilon$ is the selected divergence budget and $D$ is a selectable $f$-divergence, taken here to be either $d_{TV}$ or $D_{KL}$. Equivalently,
the embedding step moves as much probability mass as permitted by the
constraint onto the token coloured $Z_t$, while
keeping the watermarked distribution within distance $\varepsilon$ of the base
model distribution.

Because the cost is an indicator, the objective in \eqref{eq:simplemark_embedding} depends on $P$ only through the mass it places on the target vertex. Writing $B = \sum_{x : \sigma_t(x) = Z_t} Q_t(x)$ and $s = \sum_{x : \sigma_t(x) = Z_t} P(x)$, the data-processing inequality gives $D(P,Q_t) \ge D\!\left(\mathrm{Bern}(s),\mathrm{Bern}(B)\right)$.
Equality is attained by preserving the conditional distributions within the target and non-target bins. The optimum of \eqref{eq:simplemark_embedding} is therefore the two-scalar tilt
\begin{equation}
    Q_t^\star(x)
    =
    \begin{cases}
        \alpha_t\, Q_t(x), & \sigma_t(x) = Z_t, \\
        \beta_t\, Q_t(x),  & \sigma_t(x) \neq Z_t,
    \end{cases}
    \qquad
    \alpha_t = \frac{s^\star}{B},
    \quad
    \beta_t = \frac{1-s^\star}{1-B},
    \label{eq:simplemark_tilt}
\end{equation}
where $s^\star = B + \min(\varepsilon, 1-B)$ when $D = d_{TV}$, and $s^\star$ is the unique root in $[B,1]$ of $s\log\frac{s}{B} + (1-s)\log\frac{1-s}{1-B} = \varepsilon$ when $D = D_{KL}$.\footnote{For $\varepsilon \leq \log(1/B)$. For $\varepsilon>\log(1/B)$,
the constrained optimum is simply $s^* = 1$.} In certain settings, it is beneficial to limit the tilt towards unwanted tokens by capping $\alpha$.

Finally, the next token is sampled according to
$X_t \sim Q_t^\star(\cdot \mid x_{<t}).$
The generated token is appended to the context, and the procedure is repeated
for $t=1,\ldots,n$.

\paragraph{Message Decoding}
\simplemark~is decoded with a likelihood-based decoder that uses the
watermarked channel law.

Given the prompt and received token sequence $x_{1:n}$, the decoder replays the
language model to recover the base next-token distributions
$Q_t(\cdot)=Q(\cdot \mid x_{<t})$. It also reconstructs the shared key sequence
$S_t=(V_t,\sigma_t)$. For each possible code symbol
$a \in \{0,\ldots,q-1\}$, the decoder forms the corresponding target vertex
\(
    z_t(a)
    =
    \left(
        a
        +
        V_t
    \right)
    \bmod q .
\)
Using this vertex, the decoder evaluates the probability assigned to the
observed token $x_t$ by the divergence-constrained embedding distribution:
\[
    Q_{t,a}^\star(x)
    =
    \begin{cases}
        \dfrac{s^\star_a}{B_a}\, Q_t(x), & \sigma_t(x) = z_t(a), \\[6pt]
        \dfrac{1-s^\star_a}{1-B_a}\, Q_t(x),  & \sigma_t(x) \neq z_t(a),
    \end{cases}
    \qquad
    B_a = \!\!\sum_{x : \sigma_t(x) = z_t(a)}\!\! Q_t(x),
\]
with $s^\star_a$ obtained from $B_a$ and $\varepsilon$ as in the embedding step.
This defines the per-symbol log-likelihood
$L_t(a) = \log Q_{t,a}^\star(x_t).$

For a candidate message $\hat m$ with codeword
$C(m)=(C_1,\ldots,C_n)$, the decoder sums these likelihoods:
$\Lambda(\hat m) = \sum_{t=1}^n L_t(C_t( \hat m)).$
The decoded message is the maximum-likelihood codeword,
$\hat M = \arg\max_{\hat m} \Lambda( \hat m).$
Simply, the decoder selects the message whose codeword is most likely to explain the observed
token sequence under the same divergence-constrained channel used for embedding.

In Appendices~\ref{sec:encoder-decoder} and~\ref{sec:capacity}, we derive more extensive theoretical results about our watermarking scheme, 
the watermark channel family, encoding and decoding capabilities, embedding capacity, and perplexity. We discuss a few of
those results here. We acknowledge the simplicity of our results; deriving rate under a non-memoryless conditional is much more complicated and remains interesting future work. Most of our theoretical results are applicable to all $f$-divergences compatible with \simplemark.

\subsection{Perfect Distortion-Freeness}

We show that for perfect distortion freeness, that is when the watermarked distribution is equivalent to that of the unwatermarked, the watermark's capacity is zero. This motivates having an $\varepsilon$-budget that is a measurable $f$-divergence distance from the unwatermarked distribution.

\begin{proposition}[Rate vanishes as $\varepsilon\to 0$]
\label{prop:rate-vanishes}
Let $\mathcal{X}$ and $\mathcal{Z}$ be finite alphabets, let $Q\in\Delta(\mathcal{X})$ be a fixed base distribution, and let $P_Z$ be a fixed distribution on $\mathcal{Z}$. For $\varepsilon\ge 0$, define the feasible set
\[
\mathcal{F}(\varepsilon)
:=
\left\{
\bigl(P_{X\mid Z=z}\bigr)_{z\in\mathcal{Z}}
:
d_{\mathrm{TV}}(P_{X\mid Z=z},Q)\le \varepsilon
\ \text{for all } z\in\mathcal{Z}
\right\}.
\]
Define the constrained capacity
$C(\varepsilon)
:=
\sup_{\bigl(P_{X\mid Z=z}\bigr)\in\mathcal{F}(\varepsilon)}
I(Z;X),$
where the mutual information is computed under the joint law
$P_{Z,X}(z,x)=P_Z(z)P_{X\mid Z=z}(x).$
Then $\lim_{\varepsilon\to 0} C(\varepsilon)=0.$
In particular, any achievable embedding rate under the uniform per-key TV constraint must approach $0$ as $\varepsilon\to 0$.
\end{proposition}

\begin{proof}
We first show $C(0)=0$. If
$\bigl(P_{X\mid Z=z}\bigr)_{z\in\mathcal{Z}}\in\mathcal{F}(0),$
then for every $z\in\mathcal{Z}$,
$d_{\mathrm{TV}}(P_{X\mid Z=z},Q)=0.$
Since total variation distance vanishes if and only if the two distributions are equal~\citep{keener2010theoretical}, it follows that
$P_{X\mid Z=z}=Q,$ for all $ z\in\mathcal{Z}.$
Hence the conditional law of $X$ does not depend on $Z$, so $X$ and $Z$ are independent. Therefore
$I(Z;X)=0.$
Taking the supremum over all feasible channels shows that $C(0)=0.$

We now prove that $C(\varepsilon)\to 0$ as $\varepsilon\to 0$. Let $(\varepsilon_n)_{n\ge 1}$ be any sequence with $\varepsilon_n\downarrow 0$. For each $n$, choose a feasible channel family
$\bigl(P^{(n)}_{X\mid Z=z}\bigr)_{z\in\mathcal{Z}}\in\mathcal{F}(\varepsilon_n)$
such that
$I_n := I_{P^{(n)}}(Z;X)\ge C(\varepsilon_n)-\frac1n.$
Because $\mathcal{X}$ and $\mathcal{Z}$ are finite, the product of simplices
$\prod_{z\in\mathcal{Z}} \Delta(\mathcal{X})$
is compact~\citep{Munkers84}. Hence, there exists~\citep{rudin1964principles} a subsequence, which we still index by $n$, such that for every $z\in\mathcal{Z}$,
$P^{(n)}_{X\mid Z=z}\to P^\star_{X\mid Z=z}$
for some limiting channel family $\bigl(P^\star_{X\mid Z=z}\bigr)_{z\in\mathcal{Z}}$.

Since each $P^{(n)}$ is feasible,
$d_{\mathrm{TV}}\!\left(P^{(n)}_{X\mid Z=z},Q\right)\le \varepsilon_n
$ for all $z.$
Letting $n\to\infty$ gives
$d_{\mathrm{TV}}\!\left(P^\star_{X\mid Z=z},Q\right)=0
$ for all $ z,$
and therefore
$P^\star_{X\mid Z=z}=Q $ for all $z\in\mathcal{Z}.$
So under the limiting channel, $X$ is independent of $Z$, and thus
$I_{P^\star}(Z;X)=0.$

Because mutual information is continuous in the joint distribution on finite alphabets~\citep{polyanskiy2025information, CoverT2006}, we have
$I_n = I_{P^{(n)}}(Z;X)\to I_{P^\star}(Z;X)=0.$
But
$C(\varepsilon_n)\le I_n+\frac1n,$
so taking limits yields
$\limsup_{n\to\infty} C(\varepsilon_n)\le 0.$
Since mutual information is always nonnegative,
$C(\varepsilon_n)\ge 0,$
hence
$\lim_{n\to\infty} C(\varepsilon_n)=0.$

Because this holds for every sequence $\varepsilon_n\downarrow 0$, it follows that
$\lim_{\varepsilon\to 0} C(\varepsilon)=0.$
This proves the claim.
\end{proof}

Ideally, a watermarked model would sample from the base distribution
itself, and each realized key would leave every next-token
distribution exactly untouched. Proposition~\ref{prop:rate-vanishes} shows
this ideal is unachievable in principle: under exact per-key
distortion-freeness the token is independent of the target vertex,
the channel collapses, and the embedding capacity is zero; that is,
no encoder or decoder can recover a single bit. Distortion is therefore
not a defect to be engineered away but the price of embedding, and we
make that price the design parameter: an explicit $\varepsilon$
divergence budget, with Theorem~\ref{thm:main-correct-converse} quantifying the
capacity it buys.

\subsection{Rate Bounds under Uniform Per-Key TV Constraints}

Subsequently, we are able to solve for several capacity results, which can be found in Appendix ~\ref{sec:capacity}. We again emphasize that our theoretical results assume a simpler setting, but nonetheless provide meaningful insight to our watermark's capacity. Let $\varepsilon$ denote a TV distance budget.
\begin{theorem}[Rate under uniform per-key TV constraint (see Theorem~\ref{thm:correct-converse})]
\label{thm:main-correct-converse}
Assume the conditional-memoryless model. Let a sequence of length-$n$ watermarking schemes satisfy the uniform per-key $\varepsilon$-TV constraint, and let the decoder error probability
\(
P_e^{(n)}:=\Pr[\widehat M\neq M]
\)
tend to zero as $n\to\infty$. If
\(
R_n:=\frac{k}{n}
\)
denotes the embedded rate in bits per token, then for every $n$,
\(
R_n
\le
\frac1n\sum_{t=1}^n \widetilde C_t(\varepsilon)
+
\frac{1+nR_n P_e^{(n)}}{n},
\)
 where $\widetilde C_t(\varepsilon)$ is the universal constrained channel capacity given in the Appendix as Definition~\ref{def:universal-Ct}. Consequently, if $(R_n)_n$ is bounded and $P_e^{(n)}\to 0$, then
\(
\limsup_{n\to\infty} R_n
\le
\limsup_{n\to\infty}\frac1n\sum_{t=1}^n \widetilde C_t(\varepsilon).
\)
\end{theorem}

\section{Experimental Results}
\label{sec:experiments}

In this section we demonstrate \simplemark's efficacy experimentally. We evaluate \simplemark's reasoning abilities, watermarked token distinguishing advantage, message and bit decoding accuracy, and perplexity. To baseline our performance, we compare our results to \arcmark \citep{gilani2026arcmarkmultibitllmwatermark} and BiMark \citep{feng2025bimark}. Both watermarks have publicly available implementations, are state of the art for message recovery, and have some definition of distortion freeness the watermark adheres to. We use 300 trials per point unless specified.

\subsection{Experiment Parameters}

\paragraph{Setup.}
We experiment with Qwen3.5-9B-Base \citep{qwen3.5} and Llama3.1-8B-Instruct \citep{grattafiori2024llama}. When applicable, we use the C4-RealNewslike dataset \citep{raffel2020exploring} and a temperature of 1.0. Each trial generates 300 tokens and is evaluated at the prefix of $50, 100, \dots, 300$ tokens. The reasoning experiments use the instruction-tuned Qwen3.5-9B in thinking mode with temperature 0.6. \simplemark uses a field size of $q=53$, and subsequently top-53 sampling. The other baselines use top-50 sampling. The $\varepsilon$ budget is a per-step TV or KL constraint. We report results for $\varepsilon \in \{0.01, 0.05, 0.1, 0.15\}$ and, where stated, with the per-token likelihood-ratio cap $\alpha_{\max} \in \{500, 1500, 2000\}$. BiMark~\citep{feng2025bimark} is run with $d = 10$ partition layers, $\tilde{\delta} = 1.0$, context width 2 and top-50 sampling; ArcMark~\citep{gilani2026arcmarkmultibitllmwatermark} with top-50 sampling, context width 2, $r = 64$ and follow its authors' setup. All schemes see the same prompts.
For the key derivation, the encoder and decoder share a fixed integer secret key. At each token position $t$, the side information is the SHA-256 hash of the secret key concatenated with the ids of the $h$ preceding generated tokens.

\subsection{Empirical Maximum Token Distinguishing Advantage}

We measure the maximum advantage an adversary would have distinguishing between a watermarked and an unwatermarked distribution. As defined in Section \ref{sec:prelims}, total variation distance exactly measures this metric. To empirically estimate TV distance, we prompt the model with 10 prompts and record the model's unwatermarked and watermarked output 10,000 times per prompt, collect the next token output and estimate a distribution. We then calculate the TV distance between the estimated distributions. The results are shown in Figure \ref{fig:TV}. The outputs have slight noise, as even an unwatermarked output picks up a slight TV distance when estimated based on token outputs.

\begin{figure}
    \centering
    \includegraphics[width=0.9\linewidth]{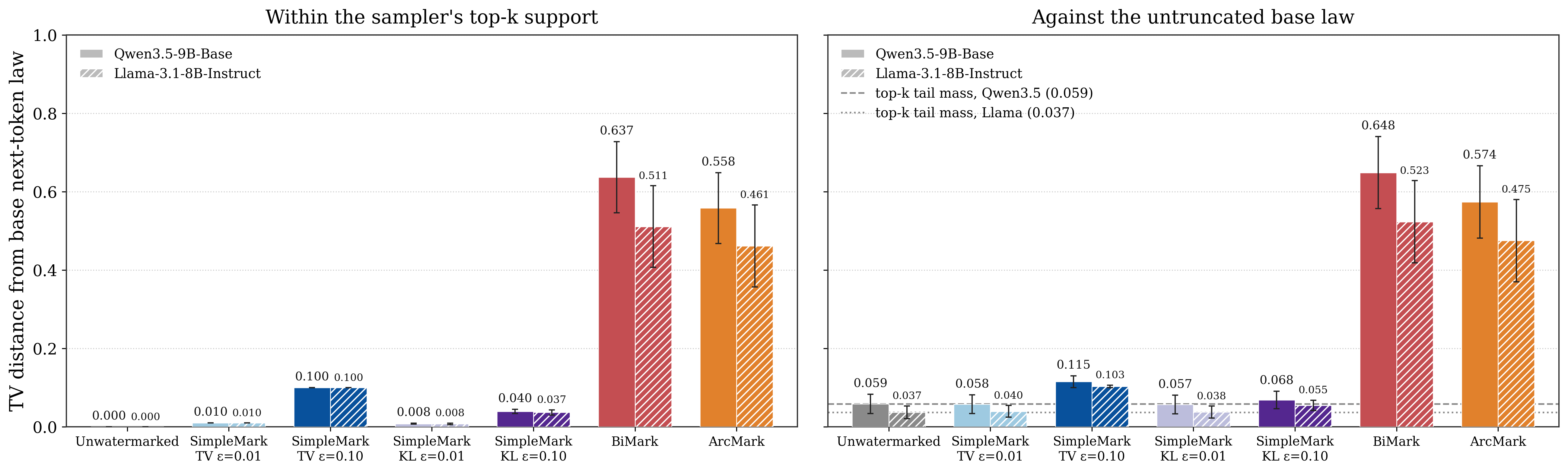}
    \caption{Empirical total variation distance and distinguishing bound from an unwatermarked output and a watermarked output for a baseline, \simplemark, Bimark, and \arcmark. The full bars represent Qwen3.5-9B-Base and the dashed bars represent Llama3.1-8B-Instruct.}
    \label{fig:TV}
\end{figure}

Interpreting these results, the lower the TV distance, the lower the chance the best possible distinguisher is able to tell whether an output was watermarked or unwatermarked. From our estimated distributions based on the model's outputs, before accounting for noise, for \simplemark~the distinguisher only gains an advantage of at most $\varepsilon$, adhering to the budget, while for \arcmark, the distinguisher gains a 46--57\% advantage and a 51--65\% advantage for Bimark. The untruncated law incurs additional TV distance from removing the tail mass outside the top-$k$. We can take away that in an adversarial setting, such as those used to analyze cryptographic protocols, other current watermarking distortion freeness metrics may not be feasible for watermarking schemes to remain secure.

\subsection{\simplemark Reasoning Performance}
We evaluate \simplemark under several reasoning tasks, GSM8K \citep{cobbe2021training}, GPQA-Diamond \citep{rein2023gpqa}, LiveCodeBench \citep{jain2025livecodebench}, and AIME~2024 \citep{aime24}. 
We use 300, 198, 100, and 30 problems respectively, with a generation budget of 4{,}096 tokens for GSM8K, 8{,}192 for GPQA-Diamond and LiveCodeBench and 16{,}384 for AIME, and every scheme is run on the identical set of problems with one sample per problem. \simplemark~embeds a 3- or 8-bit message at $\varepsilon{=}0.05$ (TV) and $\varepsilon{=}0.05$ (KL), both with $\alpha_{\max}{=}1500$. On these reasoning traces the cap binds on roughly 91\% of tokens.

Table \ref{tab:reasoning_task_accuracy} shows that by embedding \simplemark, the language model's reasoning performance does not degrade. By capping $\alpha$, we are more likely to avoid sampling undesired tokens that would degrade a model's performance. In Figure~\ref{fig:v_len_prefix_dec} we show that we can reliably recover the embedded watermark at a competitive rate under the chosen parameters. The results for Llama3.1-8B-Instruct can be found in Table~\ref{tab:reasoning_task_accuracy_llama}.

\begin{table}[t]
\centering
\small
\setlength{\tabcolsep}{5pt}
\resizebox{\linewidth}{!}{%
\begin{tabular}{lcccc}
\toprule
Scheme & GSM8K & GPQA-D & LCB & AIME 2024 \\
(trials) & (300) & (198) & ($100$) & ($30$) \\
\midrule
\multicolumn{5}{l}{\textit{3-bit payload}} \\
SimpleMark TV ($\varepsilon{=}0.05$, $\alpha_{\max}{=}1500$) & 278 & 112 & 33 & 8 \\
SimpleMark KL ($\varepsilon{=}0.05$, $\alpha_{\max}{=}1500$) & 274 & 108 & 40 & 9 \\
BiMark & 270 & 107 & 33 & 11 \\
ArcMark & 272 & 111 & 35 & 10 \\
Unwatermarked & 272 & 110 & 36 & 8 \\
\midrule
\multicolumn{5}{l}{\textit{8-bit payload}} \\
SimpleMark TV ($\varepsilon{=}0.05$, $\alpha_{\max}{=}1500$) & 275 & 103 & 38 & 9 \\
SimpleMark KL ($\varepsilon{=}0.05$, $\alpha_{\max}{=}1500$) & 273 & 102 & 41 & 8 \\
BiMark & 279 & 94 & 33 & 8 \\
ArcMark & 274 & 109 & 40 & 7 \\
Unwatermarked & 272 & 110 & 36 & 8 \\
\bottomrule
\end{tabular}}
\caption{Reasoning performance on tasks for Qwen3.5-9B (thinking mode, temperature 0.6) on GSM8K (300 problems), GPQA-Diamond (198 problems), LiveCodeBench (100 problems, graded on public test cases only, so counts are only comparable across schemes) and AIME 2024 (30 problems). We use identical tasks and datasets for every benchmark. We evaluate SimpleMark at two operating points ($q{=}53$, TV or KL $\varepsilon = 0.05$, cap $\alpha_{\max}{=}1500$), BiMark, ArcMark, and no watermark.}
\label{tab:reasoning_task_accuracy}
\end{table}

\subsection{Perplexity}

\begin{figure}
    \centering
    \includegraphics[width=0.9\linewidth]{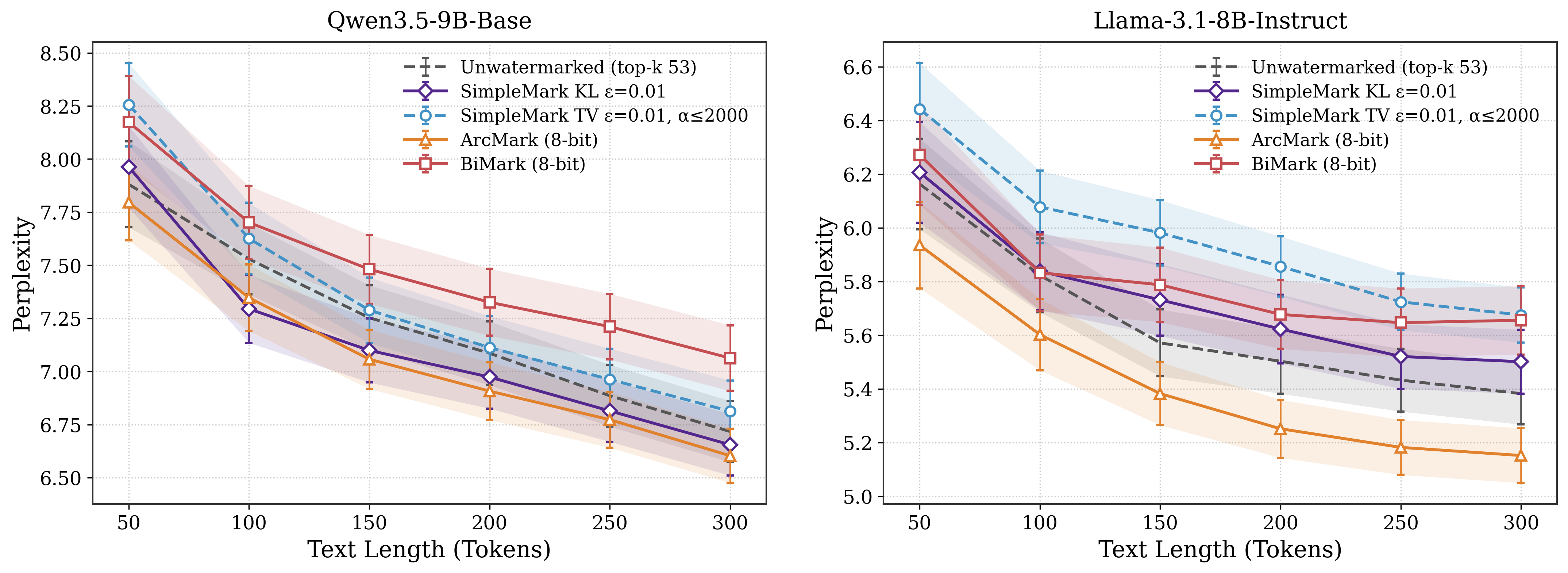}
    \caption{Perplexity results for token outputs ranging from 50-300 tokens for \simplemark, \arcmark, and Bimark. A low KL $\varepsilon$-budget \simplemark is consistently closest to that of the unwatermarked output.}
    \label{fig:perplexity}
\end{figure}

Perplexity is often used as a text quality metric for a watermark. In simple terms, it is how surprised a model is by its output through an exponentiation of the distribution's joint entropy. Perplexity can also be expressed as an exponentiation of the
watermarked generating distribution's entropy, plus the KL divergence between the unwatermarked and watermarked distributions. As expected Figure~\ref{fig:perplexity} validates that \simplemark achieves the perplexity closest to that of the unwatermarked distribution under low KL $\varepsilon$-budgets.

\subsection{Watermark Extraction Accuracy}
We measure message accuracy when recovering our watermark. Bit errors when recovering watermarks can cause the messages to be completely different, so we focus on message accuracy. For medium $\varepsilon$-budgets of $\{0.05, 0.1, 0.15\}$ we generate 300 tokens and use a prefix in increments of 50 tokens up to 300 tokens and attempt to extract the watermark. For lower $\varepsilon$-budgets of $\{0.01, 0.05\}$, we sample from 500--1000 tokens. From Figure~\ref{fig:v_len_prefix_dec}, at higher $\varepsilon$-budgets \simplemark's decoding capabilities are competitive with the current state of the art watermarks, which are designed to optimize decoding capabilities and thus have high distortion in token distributions when compared to the unwatermarked law.

\begin{figure}
    \centering
    \includegraphics[width=0.9\linewidth]{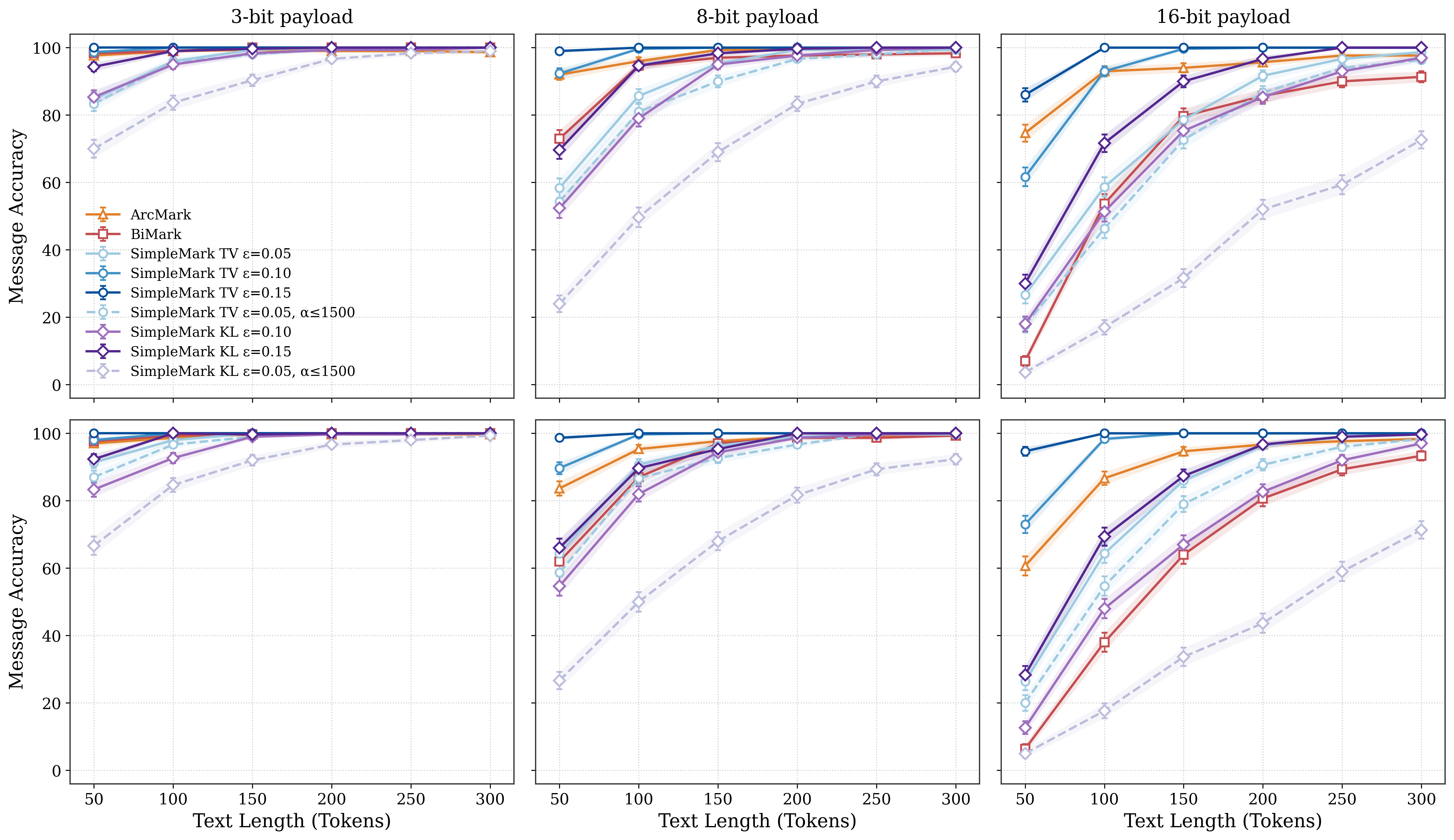}
    \caption{Message recovery for a variable length prefix of a 300 token sequence. For $\varepsilon \in \{0.05, 0.1, 0.15\}$. The \simplemark~decoding capabilities outperform the current SOTA at $\varepsilon = 0.15, $ and remain competitive with lower, more practical $\varepsilon$ budgets.}
    \label{fig:v_len_prefix_dec}
\end{figure}

Overall, our experiments validate that \simplemark achieves minimal distortion from the unwatermarked distribution as seen through the distinguishing advantage and perplexity experiments, while retaining a model's reasoning capabilities and competitive extraction accuracy.

\section{Conclusion}
\label{sec:conclusion}
In this work, 
we studied multi-bit text watermarking as a \emph{constrained communication problem}, highlighting a fundamental tradeoff between payload, reliability, and statistical detectability. Unlike average-key distortion guarantees, our per-key $f$-divergence formulation directly controls how much each realized watermark distribution can deviate from the base model, in the worst case. We show that exact per-key distortion-freeness forces zero embedding capacity, while nonzero divergence budgets lead to a tunable
rate--detectability tradeoff governed by the induced watermark channel.
Motivated by this framework, we introduced~\simplemark, which combines error-correcting codes with a simple divergence-constrained probability tilt to obtain a closed-form embedding rule. Empirically,~\simplemark substantially reduces distributional deviation from the base model while maintaining competitive message recovery, perplexity, and reasoning performance. Our results provide an information-theoretic foundation for designing multi-bit watermarks with explicit, tunable statistical-security guarantees.

\section{Acknowledgements}
\label{sec:acks}
This work was supported by a grant from Amazon via the
Amazon-Illinois Center on AI for Interactive Conversational Experiences (AICE) at 
the University of Illinois at Urbana-Champaign (UIUC).
Benjamin D. Kim is supported by an MIT Jacobs Presidential Fellowship and the National Science Foundation, under NSF GRFP award number 2141064.

Any opinions, fundings, and conclusions or recommendations expressed in this material are those of the
author(s) and do not necessarily reflect the views of the National Science Foundation or Amazon.

\clearpage

\bibliographystyle{alpha}

\bibliography{references}

@article{gilani2026arcmarkmultibitllmwatermark,
  title={ArcMark: Distortion-Free Multi-Byte LLM Watermark via Optimal Transport},
  author={Gilani, Atefeh and Vithana, Sajani and Long, Carol Xuan and Kosut, Oliver and Sankar, Lalitha and Calmon, Flavio P},
  journal={arXiv preprint arXiv:2602.07235},
  year={2026}
}

@book{gut2013probability,
  title={Probability: A Graduate Course},
  author={Gut, A.},
  isbn={9781461447078},
  lccn={2012941281},
  series={Springer Texts in Statistics},
  year={2013},
  publisher={Springer New York}
}

@book{MitzenmacherUp05,
 author = {Mitzenmacher, Michael and Upfal, Eli},
 title = {Probability and Computing: Randomized Algorithms and Probabilistic Analysis},
 year = {2005},
 isbn = {0521835402},
 publisher = {Cambridge University Press},
 address = {New York, NY, USA},
 URL = {https://doi.org/10.1017/CBO9780511813603},
 edition = {2nd}
}

@book{ross98,
  author = {Ross, Sheldon M.},
  edition = {Fifth},
  title = {A First Course in Probability},
  year = 1998
}

@book{keener2010theoretical,
  author = {Keener, R.W.},
  description = {Theoretical Statistics: Topics for a Core Course - Robert W. Keener - Google Books},
  isbn = {9780387938394},
  lccn = {2010935925},
  publisher = {Springer New York},
  series = {Springer Texts in Statistics},
  title = {Theoretical Statistics: Topics for a Core Course},
  year = 2010
}

@inproceedings{KirchenbauerGWK23,
  author       = {John Kirchenbauer and
                  Jonas Geiping and
                  Yuxin Wen and
                  Jonathan Katz and
                  Ian Miers and
                  Tom Goldstein},
  title        = {A Watermark for Large Language Models},
  booktitle    = {International Conference on Machine Learning, {ICML} 2023, 23-29 July
                  2023, Honolulu, Hawaii, {USA}},
  series       = {Proceedings of Machine Learning Research},
  pages        = {17061--17084},
  publisher    = {{PMLR}},
  year         = {2023},
  url          = {https://proceedings.mlr.press/v202/kirchenbauer23a.html}
}

@book{CoverT2006,
  at = {2008-03-31 06:17:47},
  author = {Cover, Thomas M. and Thomas, Joy A.},
  isbn = {0471241954},
  month = {July},
  publisher = {Wiley-Interscience},
  title = {Elements of Information Theory 2nd Edition (Wiley Series in Telecommunications and Signal Processing)},
  year = 2006
}

@book{polyanskiy2025information,
  title={Information Theory: From Coding to Learning},
  author={Polyanskiy, Y. and Wu, Y.},
  isbn={9781108832908},
  lccn={2024004713},
  year={2025},
  publisher={Cambridge University Press}
}

@article{MoulinO03,
  author       = {Pierre Moulin and
                  Joseph A. O'Sullivan},
  title        = {Information-theoretic analysis of information hiding},
  journal      = {{IEEE} Trans. Inf. Theory},
  volume       = {49},
  number       = {3},
  pages        = {563--593},
  year         = {2003},
  url          = {https://doi.org/10.1109/TIT.2002.808134},
  doi          = {10.1109/TIT.2002.808134}
}

@INPROCEEDINGS{Moulin05,
  author={Moulin, Pierre},
  booktitle={2005 13th European Signal Processing Conference}, 
  title={The method of types and its application to information hiding}, 
  year={2005},
  volume={},
  number={},
  pages={1-1},
  doi={}}

@InProceedings{Moulin03,
author="Moulin, Pierre",
title="Information-Hiding Games",
booktitle="Digital Watermarking",
year="2003",
publisher="Springer Berlin Heidelberg",
address="Berlin, Heidelberg",
pages="1--12",
isbn="978-3-540-36617-1"
}

@article{MoulinK2005,
    author  = "Pierre Moulin and Ralf Koetter",
    title   = "Data-Hiding Codes",
    journal = "Proceedings of the IEEE",
    volume  = "93",
    number  = "12",
    month   = dec,
    year    = "2005",
    pages   = "2083--2126"
}

@inproceedings{zhu2026quantilemark,
  title={QuantileMark: A Message-Symmetric Multi-bit Watermark for LLMs},
  author={Zhu, Junlin and Huang, Baizhou and Wan, Xiaojun},
  booktitle={Proceedings of the 64th Annual Meeting of the Association for Computational Linguistics (Volume 1: Long Papers)},
  pages={6790--6806},
  year={2026}
}

@inproceedings{feng2025bimark,
title={BiMark: Unbiased Multilayer Watermarking for Large Language Models},
author={Xiaoyan Feng and He Zhang and Yanjun Zhang and Leo Yu Zhang and Shirui Pan},
booktitle={Forty-second International Conference on Machine Learning},
year={2025},
url={https://openreview.net/forum?id=Zvyb3WAg03}
}

@inproceedings{yoo-etal-2024-advancing,
    title = "Advancing Beyond Identification: Multi-bit Watermark for Large Language Models",
    author = "Yoo, KiYoon  and
      Ahn, Wonhyuk  and
      Kwak, Nojun",
    editor = "Duh, Kevin  and
      Gomez, Helena  and
      Bethard, Steven",
    booktitle = "Proceedings of the 2024 Conference of the North American Chapter of the Association for Computational Linguistics: Human Language Technologies (Volume 1: Long Papers)",
    month = jun,
    year = "2024",
    address = "Mexico City, Mexico",
    publisher = "Association for Computational Linguistics",
    url = "https://aclanthology.org/2024.naacl-long.224/",
    doi = "10.18653/v1/2024.naacl-long.224",
    pages = "4031--4055",
}

@InProceedings{pmlr-v267-jiang25j,
  title = 	 {{S}tealth{I}nk: A Multi-bit and Stealthy Watermark for Large Language Models},
  author =       {Jiang, Ya and Wu, Chuxiong and Boroujeny, Massieh Kordi and Mark, Brian and Zeng, Kai},
  booktitle = 	 {Proceedings of the 42nd International Conference on Machine Learning},
  pages = 	 {27685--27709},
  year = 	 {2025},
  volume = 	 {267},
  series = 	 {Proceedings of Machine Learning Research},
  month = 	 {13--19 Jul},
  publisher =    {PMLR},
  url = 	 {https://proceedings.mlr.press/v267/jiang25j.html}
}

@inproceedings{xu2025robust,
title={Robust Multi-bit Text Watermark with {LLM}-based Paraphrasers},
author={Xiaojun Xu and Jinghan Jia and Yuanshun Yao and Yang Liu and Hang Li},
booktitle={Forty-second International Conference on Machine Learning},
year={2025},
url={https://openreview.net/forum?id=DVjkling5x}
}

@inproceedings{he2025theoretically,
title={Theoretically Grounded Framework for {LLM} Watermarking: A Distribution-Adaptive Approach},
author={Haiyun He and Yepeng Liu and Ziqiao Wang and Yongyi Mao and Yuheng Bu},
booktitle={The Thirty-ninth Annual Conference on Neural Information Processing Systems},
year={2025},
url={https://openreview.net/forum?id=CMmKcHFDKL}
}

@article{CW06,
  title={Quantization index modulation: A class of provably good methods for digital watermarking and information embedding},
  author={Chen, Brian and Wornell, Gregory W},
  journal={IEEE Transactions on Information theory},
  volume={47},
  number={4},
  pages={1423--1443},
  year={2001},
  publisher={IEEE}
}

@article{CACHIN200441,
title = {An information-theoretic model for steganography},
journal = {Information and Computation},
volume = {192},
number = {1},
pages = {41-56},
year = {2004},
issn = {0890-5401},
doi = {https://doi.org/10.1016/j.ic.2004.02.003},
url = {https://www.sciencedirect.com/science/article/pii/S0890540104000409},
author = {Christian Cachin},
}

@book{CMBFK07,
author = {Cox, Ingemar and Miller, Matthew and Bloom, Jeffrey and Fridrich, Jessica and Kalker, Ton},
title = {Digital Watermarking and Steganography},
year = {2007},
isbn = {9780080555805},
publisher = {Morgan Kaufmann Publishers Inc.},
address = {San Francisco, CA, USA},
edition = {2},
}

@inproceedings{qu2025provably,
  title={Provably robust multi-bit watermarking for $\{$AI-generated$\}$ text},
  author={Qu, Wenjie and Zheng, Wengrui and Tao, Tianyang and Yin, Dong and Jiang, Yanze and Tian, Zhihua and Zou, Wei and Jia, Jinyuan and Zhang, Jiaheng},
  booktitle={34th USENIX Security Symposium (USENIX Security 25)},
  pages={201--220},
  year={2025}
}

@inproceedings{cohen2025watermarking,
  title={Watermarking language models for many adaptive users},
  author={Cohen, Aloni and Hoover, Alexander and Schoenbach, Gabe},
  booktitle={2025 IEEE Symposium on Security and Privacy (SP)},
  pages={2583--2601},
  year={2025},
  organization={IEEE}
}

@article{xu2026xmark,
  title={XMark: Reliable Multi-Bit Watermarking for LLM-Generated Texts},
  author={Xu, Jiahao and Hu, Rui and Kotevska, Olivera and Zhang, Zikai},
  journal={arXiv preprint arXiv:2604.05242},
  year={2026}
}

@article{Longpre24,
	author = {Longpre, Shayne and Mahari, Robert and Chen, Anthony and Obeng-Marnu, Naana and Sileo, Damien and Brannon, William and Muennighoff, Niklas and Khazam, Nathan and Kabbara, Jad and Perisetla, Kartik and Wu, Xinyi (Alexis) and Shippole, Enrico and Bollacker, Kurt and Wu, Tongshuang and Villa, Luis and Pentland, Sandy and Hooker, Sara},
	doi = {10.1038/s42256-024-00878-8},
	isbn = {2522-5839},
	journal = {Nature Machine Intelligence},
	number = {8},
	pages = {975--987},
	title = "A large-scale audit of dataset licensing and attribution in {AI}",
	url = {https://doi.org/10.1038/s42256-024-00878-8},
	volume = {6},
	year = {2024}
}

@article{rein2023gpqa,
  title={Gpqa: A graduate-level google-proof q\&a benchmark},
  author={Rein, David and Hou, Betty Li and Stickland, Asa Cooper and Petty, Jackson and Pang, Richard Yuanzhe and Dirani, Julien and Michael, Julian and Bowman, Samuel R},
  journal={arXiv preprint arXiv:2311.12022},
  year={2023}
}

@article{cobbe2021training,
  title={Training verifiers to solve math word problems},
  author={Cobbe, Karl and Kosaraju, Vineet and Bavarian, Mohammad and Chen, Mark and Jun, Heewoo and Kaiser, Lukasz and Plappert, Matthias and Tworek, Jerry and Hilton, Jacob and Nakano, Reiichiro and others},
  journal={arXiv preprint arXiv:2110.14168},
  year={2021}
}

@book{bok:MW,
  author = {MacWilliams, F.J. and Sloane, N.J.A.},
  edition = {2nd},
  publisher = {North-holland Publishing Company},
  title = {The Theory of Error-Correcting Codes},
  year = 1978
}

@article{grattafiori2024llama,
  title={The llama 3 herd of models},
  author={Grattafiori, Aaron and Dubey, Abhimanyu and Jauhri, Abhinav and Pandey, Abhinav and Kadian, Abhishek and Al-Dahle, Ahmad and Letman, Aiesha and Mathur, Akhil and Schelten, Alan and Vaughan, Alex and others},
  journal={arXiv preprint arXiv:2407.21783},
  year={2024}
}

@misc{qwen3.5,
    title  = {{Qwen3.5}: Towards Native Multimodal Agents},
    author = {{Qwen Team}},
    month  = {February},
    year   = {2026},
    url    = {https://qwen.ai/blog?id=qwen3.5}
}

@book{Munkers84,
  author = {Munkres, James R.},
  howpublished = {Hardcover},
  publisher = {Addison Wesley Publishing Company},
  title = {{Elements of Algebraic Topology}},
  year = 1984
}

@book{rudin1964principles,
  title={Principles of Mathematical Analysis},
  author={Rudin, W.},
  isbn={9780070542310},
  lccn={63021479},
  series={International series in pure and applied mathematics},
  url={https://books.google.com/books?id=yifvAAAAMAAJ},
  year={1964},
  publisher={McGraw-Hill}
}

@article{ChristGZ23,
  author       = {Miranda Christ and
                  Sam Gunn and
                  Or Zamir},
  title        = {Undetectable Watermarks for Language Models},
  journal      = {{IACR} Cryptol. ePrint Arch.},
  pages        = {763},
  year         = {2023},
  url          = {https://eprint.iacr.org/2023/763}
}

@article{kuditipudi2024robust,
title={{Robust distortion-free watermarks for language models}},
author={Rohith Kuditipudi and John Thickstun and Tatsunori Hashimoto and Percy Liang},
journal={Transactions on Machine Learning Research},
year={2024},
url={https://openreview.net/forum?id=FpaCL1MO2C},
}

@article{jiang2026mirrormark,
  title={MirrorMark: Generalizable Mirrored Sampling for Multi-bit LLM Watermarking},
  author={Jiang, Ya and Boroujeny, Massieh Kordi and Kumar, Surender Suresh and Zeng, Kai},
  journal={arXiv preprint arXiv:2601.22246},
  year={2026}
}

@misc{cui2026mc,
      title={MC$^2$Mark: Distortion-Free Multi-Bit Watermarking for Long Messages}, 
      author={Xuehao Cui and Ruibo Chen and Yihan Wu and Heng Huang},
      year={2026},
      eprint={2602.14030},
      archivePrefix={arXiv},
      primaryClass={cs.CR},
      url={https://arxiv.org/abs/2602.14030}, 
}

@inproceedings{jain2025livecodebench,
  title={Livecodebench: Holistic and contamination free evaluation of large language models for code},
  author={Jain, Naman and Gu, Alex and Li, Wen-Ding and Yan, Fanjia and Zhang, Tianjun and Wang, Sida and Solar-Lezama, Armando and Sen, Koushik and Stoica, Ion},
  booktitle={International Conference on Learning Representations},
  volume={2025},
  pages={58791--58831},
  year={2025}
}

@misc{aime24,
      title={American Invitational Mathematics Examination (AIME) 2024}, 
      author={Zhang, Yifan and Math-AI, Team},
      year={2024},
}

@book{KL14,
author = {Katz, Jonathan and Lindell, Yehuda},
title = {Introduction to Modern Cryptography, Second Edition},
year = {2014},
isbn = {1466570261},
publisher = {Chapman \& Hall/CRC},
edition = {2nd}
}

@article{raffel2020exploring,
  title={Exploring the limits of transfer learning with a unified text-to-text transformer},
  author={Raffel, Colin and Shazeer, Noam and Roberts, Adam and Lee, Katherine and Narang, Sharan and Matena, Michael and Zhou, Yanqi and Li, Wei and Liu, Peter J},
  journal={Journal of machine learning research},
  volume={21},
  number={140},
  pages={1--67},
  year={2020}
}

@misc{guo2026feedback,
      title={Feedback Coding Enables Inference-Time Covert Agentic Communication}, 
      author={Sidong Guo and Sajani Vithana and Atefeh Gilani and Lalitha Sankar and Oliver Kosut and Flavio P. Calmon},
      year={2026},
      eprint={2609.24994},
      archivePrefix={arXiv},
      primaryClass={cs.IT},
      url={https://arxiv.org/abs/2609.24994}, 
}

@inproceedings{
      zhang2024watermarks,
      title={Watermarks in the Sand: Impossibility of Strong Watermarking for Generative Models},
      author={Hanlin Zhang and Benjamin L. Edelman and Danilo Francati and Daniele Venturi and Giuseppe Ateniese and Boaz Barak},
      booktitle={Forty-first International Conference on Machine Learning},
      year={2024},
      }

@inproceedings{AKV03,
author = {Adelsbach, Andr{\'e} and Katzenbeisser, Stefan and Veith, Helmut},
title = {Watermarking schemes provably secure against copy and ambiguity attacks},
year = {2003},
isbn = {1581137869},
publisher = {Association for Computing Machinery},
address = {New York, NY, USA},
url = {https://doi.org/10.1145/947380.947395},
doi = {10.1145/947380.947395},
booktitle = {Proceedings of the 3rd ACM Workshop on Digital Rights Management},
pages = {111–119},
numpages = {9},
location = {Washington, DC, USA},
series = {DRM '03}
}

@inproceedings{chang-etal-2025-watermark,
    title = "Watermark Smoothing Attacks against Language Models",
    author = "Chang, Hongyan  and
      Hassani, Hamed  and
      Shokri, Reza",
    editor = "Christodoulopoulos, Christos  and
      Chakraborty, Tanmoy  and
      Rose, Carolyn  and
      Peng, Violet",
    booktitle = "Findings of the Association for Computational Linguistics: EMNLP 2025",
    month = nov,
    year = "2025",
    address = "Suzhou, China",
    publisher = "Association for Computational Linguistics",
    url = "https://aclanthology.org/2025.findings-emnlp.264/",
    doi = "10.18653/v1/2025.findings-emnlp.264",
    pages = "4915--4941",
    ISBN = "979-8-89176-335-7",
}

\clearpage
\appendix

\newpage

\subsection*{AI use statement}

In this work, we used generative AI tools for generating the code 
to run the experiments.
We also used generative AI to generate Figure~\ref{fig:diagram}
that explains our scheme on a high level.

We have reviewed all AI-assisted work. 
We take full responsibility for the final content of this work.

\section{Encoder and Decoder Algorithms}
\label{sec:encoder-decoder}

\simplemark's encoding and decoding procedures are described in Section~\ref{sec:main-algorithms}. In this section, we share the algorithmic forms of Algorithm~\ref{alg:pk-tv-encoder} (encoder) and Algorithm~\ref{alg:pk-tv-decoder} (decoder).

Algorithm~\ref{alg:pk-tv-encoder} describes the encoder: the message is
LDPC-encoded once, and each generation step derives the key from the
preceding tokens. Subsequently the corresponding target vertex and tilt towards that token are calculated. A token is then sampled from the tilted watermarked distribution.

Algorithm~\ref{alg:pk-tv-decoder} gives the decoder: the decoder starts by recovering the base laws and key sequence from the received text. Next it scores each code symbol under the same $f$-divergence constrained channel used for
embedding, and returns the maximum-likelihood codeword.

\begin{algorithm}[t]
\caption{Per-key divergence-constrained encoder}
\label{alg:pk-tv-encoder}
\begin{algorithmic}[1]
\Require Message $m\in\mathcal{M}$, blocklength $n$, LDPC codebook $\{C_m(1:n)\}_{m\in\mathcal{M}}$ over $\mathbb{F}_q$, secret key $K$, hash window $h$, divergence $D\in\{d_{\mathrm{TV}},D_{KL}\}$, budget $\varepsilon$, base laws $(Q_t)_{t=1}^n$
\Ensure Watermarked token sequence $X_{1:n}$

\State Encode $m$ as the codeword $C_m(1),\dots,C_m(n)$.
\For{$t=1$ to $n$}
    \State Derive the side information from the key and the preceding tokens,
    \[
    (V_t,\sigma_t) \gets \mathrm{Hash}(K,\,X_{t-h:t-1}),
    \]
    where $V_t\in\mathbb{F}_q$ and $\sigma_t:\mathcal{S}_t\to\mathbb{F}_q$ is a keyed bijection on the truncated support.
    \State Form the target vertex
    \[
    Z_t \gets (C_m(t)+V_t) \bmod q.
    \]
    \State Compute the base mass on the target vertex,
    \[
    B_t \gets \sum_{x\,:\,\sigma_t(x)=Z_t} Q_t(x),
    \]
    and obtain $s^\star_t$ from $(B_t,\varepsilon,D)$: $s^\star_t = B_t+\min(\varepsilon,1-B_t)$ if $D=d_{\mathrm{TV}}$; the unique root in $[B_t,1]$ of $s\log\frac{s}{B_t}+(1-s)\log\frac{1-s}{1-B_t}=\varepsilon$ if $D=D_{KL}$. If applicable, cap $s^*_t = \min\{s^*, \alpha_{\max} B\}$
    \State Form the closed-form tilt
    \[
    Q_t^\star(x)
    \gets
    \begin{cases}
    \dfrac{s^\star_t}{B_t}\,Q_t(x), & \sigma_t(x)=Z_t,\\[6pt]
    \dfrac{1-s^\star_t}{1-B_t}\,Q_t(x), & \sigma_t(x)\neq Z_t.
    \end{cases}
    \]
    \State Sample
    \[
    X_t \sim Q_t^\star.
    \]
\EndFor
\State Output $X_{1:n}$; the side information $(V_t,\sigma_t)_{t=1}^n$ is recomputable from $K$ and the text itself.
\end{algorithmic}
\end{algorithm}

\begin{algorithm}[t]
\caption{Per-key divergence-constrained MLE Decoder}
\label{alg:pk-tv-decoder}
\begin{algorithmic}[1]
\Require Observed tokens $x_{1:n}$, secret key $K$, hash window $h$, codebook $\{C_m(1:n)\}_{m\in\mathcal{M}}$, divergence $D$, budget $\varepsilon$, base model $Q(\cdot\mid\cdot)$
\Ensure Message estimate $\widehat m$

\For{$t=1$ to $n$}
    \State Recover the base law $Q_t(\cdot)=Q(\cdot\mid x_{<t})$, truncated and renormalized as in encoding, and recompute $(V_t,\sigma_t)\gets\mathrm{Hash}(K,\,x_{t-h:t-1})$.
    \State For each symbol $a\in\mathbb{F}_q$, form the candidate vertex and its base mass,
    \[
    z_t(a)=(a+v_t)\bmod q,
    \qquad
    B_a \gets \sum_{x\,:\,\sigma_t(x)=z_t(a)} Q_t(x),
    \]
    obtain $s^\star_a$ from $(B_a,\varepsilon,D)$ as in Algorithm~\ref{alg:pk-tv-encoder}, and store the per-symbol log-likelihood
    \[
    L_t(a)
    \gets
    \begin{cases}
    \log\!\bigl(\tfrac{s^\star_a}{B_a}\,Q_t(x_t)\bigr), & \sigma_t(x_t)=z_t(a),\\[6pt]
    \log\!\bigl(\tfrac{1-s^\star_a}{1-B_a}\,Q_t(x_t)\bigr), & \sigma_t(x_t)\neq z_t(a).
    \end{cases}
    \]
\EndFor
\For{each candidate message $m\in\mathcal{M}$}
    \State Compute the log-likelihood score
    \[
    \Lambda(m)
    :=
    \sum_{t=1}^n L_t\bigl(C_m(t)\bigr).
    \]
\EndFor
\State Output any maximizer
\[
\widehat m \in \arg\max_{m\in\mathcal{M}} \Lambda(m).
\]
\end{algorithmic}
\end{algorithm}

\section{Theoretical Results on Capacity of Watermarking Channel}
\label{sec:capacity}

Throughout the appendix, we consider TV distance as the $f$-divergence to solve for all our theoretical results. As previously stated, some of the theoretical results are limited to simpler settings and most results hold for all $f$-divergences. We leave additional results to future work.

Also, for simplicity, we assume independent side information and a (conditional) memoryless channel.

\subsection{Token-Level}

\paragraph{Notation and standing assumptions.}
Let $\mathcal{X}=[1:N]$ denote the token vocabulary. At time $t$, the base language model induces a next-token distribution
\[
Q_t \in \Delta(\mathcal{X}),
\]
where $\Delta(\mathcal{X})$ is the probability simplex on $\mathcal{X}$.

As in \arcmark and \simplemark, the message
\[
M \in \mathcal{M} := \{1,\dots,2^k\}
\]

is encoded into a codeword with an error correcting code
\[
C_M(1),\dots,C_M(n)\in \mathbb{F}_q := \{0,1,\dots,q-1\},
\]
and the encoder and decoder share side information
\[
V_t \sim \mathrm{Unif}(\mathbb{F}_q), \qquad t=1,\dots,n.
\]
The state variable is the target vertex
\[
Z_t := (C_M(t) + V_t) \bmod q \;\in\; \mathcal{Z} := \mathbb{F}_q,
\]
identified with the vertices of the $(q-1)$-simplex via the keyed
bijection $\sigma_t$ of Section~\ref{sec:main-algorithms}.

We assume throughout that:

\begin{enumerate}
\item The message $M$ is uniform on $\mathcal{M}$.
\item The side information sequence $V_1,\dots,V_n$ is i.i.d. and independent of $M$.
\item Given $(M,V_1,\dots,V_n)$, the encoder generates tokens through a conditional-memoryless channel:
\[
P_{X_{1:n}\mid M,V_{1:n}}(x_{1:n}\mid m,v_{1:n})
=
\prod_{t=1}^n
P_{X_t\mid Z_t=z_t(m,v_t)}(x_t).
\]
\item The decoder observes $(X_{1:n},V_{1:n})$ and outputs $\widehat M$.
\end{enumerate}

\begin{definition}[Uniform per-key $\varepsilon$-TV distortion]
Fix $\varepsilon \ge 0$. A watermarking scheme is said to be \emph{uniformly per-key $\varepsilon$-TV distortion-free} if for every time $t$, every message $m\in\mathcal{M}$, and every realized key value $v\in\mathbb{F}_q$,
\[
d_{\mathrm{TV}}
\bigl(
P_{X_t\mid M=m,V_t=v},
Q_t
\bigr)
\le \varepsilon.
\]
Equivalently, since $Z_t$ is a deterministic function of $(m,v)$, the scheme is uniformly per-key $\varepsilon$-TV distortion-free iff
\[
d_{\mathrm{TV}}
\bigl(
P_{X_t\mid Z_t=z},
Q_t
\bigr)
\le \varepsilon
\qquad \forall z\in\mathcal{Z},\ \forall t.
\]
\end{definition}

\begin{definition}[Feasible per-key channel family]
For each time $t$ and distortion budget $\varepsilon\ge 0$, define
\[
\mathcal{F}_t(\varepsilon)
:=
\left\{
\bigl(P_{X_t\mid Z_t=z}\bigr)_{z\in\mathcal{Z}}
:
d_{\mathrm{TV}}(P_{X_t\mid Z_t=z},Q_t)\le \varepsilon
\ \ \forall z\in\mathcal{Z}
\right\}.
\]
\end{definition}

\begin{definition}[Per-key constrained embedding channel]
For each time $t$, choose a channel family
\[
\bigl(P^{(\varepsilon)}_{X_t\mid Z_t=z}\bigr)_{z\in\mathcal{Z}}
\in \mathcal{F}_t(\varepsilon).
\]
Given message $m$ and side information $v_t$, form
\[
z_t=z_t(m,v_t)=(C_m(t)+v_t) \bmod q,
\]
and sample
\[
X_t \sim P^{(\varepsilon)}_{X_t\mid Z_t=z_t}.
\]
\end{definition}

\begin{proposition}[Uniform per-key distortion-freeness]
\label{prop:uniform-per-key-tv}
If the encoder uses channel families
\[
\bigl(P^{(\varepsilon)}_{X_t\mid Z_t=z}\bigr)_{z\in\mathcal{Z}}
\in \mathcal{F}_t(\varepsilon)
\qquad \text{for all } t,
\]
then the resulting scheme is uniformly per-key $\varepsilon$-TV distortion-free:
\[
d_{\mathrm{TV}}
\bigl(
P_{X_t\mid M=m,V_t=v},
Q_t
\bigr)
\le \varepsilon
\qquad
\forall t,\ \forall m,\ \forall v.
\]
\end{proposition}

\begin{proof}
Fix $t\in[1:n]$, $m\in\mathcal{M}$, and $v\in\mathbb{F}_q$. By construction,
\[
Z_t = z_t(m,v) = (C_m(t)+v) \bmod q \in \mathcal{Z}.
\]
Conditioned on $(M,V_t)=(m,v)$, the encoder samples $X_t$ according to the conditional law attached to this state:
\[
P_{X_t\mid M=m,V_t=v}
=
P^{(\varepsilon)}_{X_t\mid Z_t=z_t(m,v)}.
\]
Since
\[
\bigl(P^{(\varepsilon)}_{X_t\mid Z_t=z}\bigr)_{z\in\mathcal{Z}}
\in \mathcal{F}_t(\varepsilon),
\]
the definition of $\mathcal{F}_t(\varepsilon)$ yields
\[
d_{\mathrm{TV}}
\Bigl(
P^{(\varepsilon)}_{X_t\mid Z_t=z},
Q_t
\Bigr)
\le \varepsilon
\qquad \forall z\in\mathcal{Z}.
\]
Applying this to $z=z_t(m,v)$ gives
\[
d_{\mathrm{TV}}
\Bigl(
P_{X_t\mid M=m,V_t=v},
Q_t
\Bigr)
=
d_{\mathrm{TV}}
\Bigl(
P^{(\varepsilon)}_{X_t\mid Z_t=z_t(m,v)},
Q_t
\Bigr)
\le \varepsilon.
\]
Because $t,m,v$ were arbitrary, the claim follows uniformly over all times, messages, and realized keys.
\end{proof}

\begin{definition}[Single-letter per-time constrained capacity]
For each time $t$, define
\[
C_t(\varepsilon)
:=
\sup_{\bigl(P_{X_t\mid Z_t=z}\bigr)_{z\in\mathcal{Z}}\in \mathcal{F}_t(\varepsilon)}
I(Z_t;X_t),
\]
where the mutual information is computed under the joint law
\[
P_{Z_t,X_t}(z,x)=P_{Z_t}(z)P_{X_t\mid Z_t=z}(x).
\]
\end{definition}

\begin{definition}[Universal single-letter constrained capacity]
\label{def:universal-Ct}
For each time $t$, define
\[
\widetilde C_t(\varepsilon)
:=
\sup_{\substack{(P_{X_t\mid Z_t=z})_{z\in\mathcal Z}\in \mathcal F_t(\varepsilon)\\ P_{Z_t}\in \Delta(\mathcal Z)}}
I(Z_t;X_t),
\]
where the mutual information is computed under the joint law
\[
P_{Z_t,X_t}(z,x)=P_{Z_t}(z)\,P_{X_t\mid Z_t=z}(x).
\]
\end{definition}

\begin{theorem}[Converse under uniform per-key TV constraint]
\label{thm:correct-converse}
Assume the conditional-memoryless model. Let a sequence of length-$n$ watermarking schemes satisfy the uniform per-key $\varepsilon$-TV constraint, and let the decoder error probability
\[
P_e^{(n)}:=\Pr[\widehat M\neq M]
\]
tend to zero as $n\to\infty$. If
\[
R_n:=\frac{k}{n}
\]
denotes the embedded rate in bits per token, then for every $n$,
\[
R_n
\le
\frac1n\sum_{t=1}^n \widetilde C_t(\varepsilon)
+
\frac{1+nR_n P_e^{(n)}}{n}.
\]
Consequently, if $(R_n)_n$ is bounded and $P_e^{(n)}\to 0$, then
\[
\limsup_{n\to\infty} R_n
\le
\limsup_{n\to\infty}\frac1n\sum_{t=1}^n \widetilde C_t(\varepsilon).
\]
\end{theorem}

\begin{proof}
Since $M$ is uniform on $\mathcal M$ with $|\mathcal M|=2^k$,
\[
H(M)=k=nR_n.
\]
By Fano's inequality~\citep{CoverT2006},
\[
H(M\mid X_{1:n},V_{1:n})
\le
h_2(P_e^{(n)}) + P_e^{(n)}\log(|\mathcal M|-1)
\le
1+nR_nP_e^{(n)}.
\]
Therefore,
\[
nR_n
=
H(M)
=
I(M;X_{1:n},V_{1:n}) + H(M\mid X_{1:n},V_{1:n}),
\]
so
\[
nR_n
\le
I(M;X_{1:n},V_{1:n}) + 1+nR_nP_e^{(n)}.
\]
Because $M$ and $V_{1:n}$ are independent,
\[
I(M;X_{1:n},V_{1:n})
=
I(M;X_{1:n}\mid V_{1:n}).
\]
Hence
\[
nR_n
\le
I(M;X_{1:n}\mid V_{1:n}) + 1+nR_nP_e^{(n)}.
\]
Applying the chain rule gives
\[
I(M;X_{1:n}\mid V_{1:n})
=
\sum_{t=1}^n I(M;X_t\mid X_{1:t-1},V_{1:n}).
\]

Fix $t$. Let
\[
U_t := (X_{1:t-1},V_{1:n}).
\]
Since $Z_t$ is a deterministic function of $(M,V_t)$, and by the conditional-memoryless assumption
\[
P_{X_t\mid M,V_{1:n},X_{1:t-1}}
=
P_{X_t\mid Z_t},
\]
we have the conditional Markov chain
\[
M \longrightarrow Z_t \longrightarrow X_t
\qquad \text{given } U_t.
\]
Therefore, by data processing~\citep{CoverT2006},
\[
I(M;X_t\mid U_t)\le I(Z_t;X_t\mid U_t).
\]
Thus
\[
I(M;X_{1:n}\mid V_{1:n})
\le
\sum_{t=1}^n I(Z_t;X_t\mid U_t).
\]

Now expand the conditional mutual information:
\[
I(Z_t;X_t\mid U_t)
=
\sum_{u} P_{U_t}(u)\, I(Z_t;X_t\mid U_t=u).
\]
Fix any $u$ with $P_{U_t}(u)>0$. Conditioned on $U_t=u$, the channel from $Z_t$ to $X_t$ is still the same single-letter channel family
\[
(P_{X_t\mid Z_t=z})_{z\in\mathcal Z},
\]
while the input law of $Z_t$ becomes the posterior distribution
\[
P_{Z_t\mid U_t=u}.
\]
Because the scheme satisfies the uniform per-key $\varepsilon$-TV constraint, this channel family belongs to $\mathcal F_t(\varepsilon)$. Hence, by Definition~\ref{def:universal-Ct},
\[
I(Z_t;X_t\mid U_t=u)\le \widetilde C_t(\varepsilon).
\]
Since this holds for every $u$,
\[
I(Z_t;X_t\mid U_t)\le \widetilde C_t(\varepsilon).
\]
Therefore,
\[
I(M;X_{1:n}\mid V_{1:n})
\le
\sum_{t=1}^n \widetilde C_t(\varepsilon).
\]

Substituting back into the Fano bound yields
\[
nR_n
\le
\sum_{t=1}^n \widetilde C_t(\varepsilon) + 1+nR_nP_e^{(n)}.
\]
Dividing by $n$ gives
\[
R_n
\le
\frac1n\sum_{t=1}^n \widetilde C_t(\varepsilon)
+
\frac{1+nR_nP_e^{(n)}}{n}.
\]
This proves the finite-$n$ bound.

If $(R_n)_n$ is bounded and $P_e^{(n)}\to 0$, then the additive term on the right tends to $0$, so
\[
\limsup_{n\to\infty} R_n
\le
\limsup_{n\to\infty}\frac1n\sum_{t=1}^n \widetilde C_t(\varepsilon).
\]
\end{proof}

The following achievability result should be interpreted as an achievability theorem for the
induced single-letter channel \(P_{X|Z}\), assuming that the encoder can realize codewords
\(z^n(m)\in\mathcal Z^n\) with empirical distribution approaching \(P_Z\). A separate codebook
realizability argument is needed to specialize this statement to the constrained \simplemark or \arcmark~state
process \(Z_t(m,V_t)\).

\begin{theorem}[Achievability in the idealized induced DMC model]
\label{thm:tv-achievability}
Assume now that:

\begin{enumerate}
\item $Q_t\equiv Q$ does not depend on $t$;
\item the feasible set is the same at every time, i.e.
\[
\mathcal{F}_t(\varepsilon)\equiv \mathcal{F}(\varepsilon);
\]
\item the encoder uses a fixed channel family
\[
\bigl(P_{X\mid Z=z}\bigr)_{z\in\mathcal{Z}}\in \mathcal{F}(\varepsilon)
\]
at every time;
\item the induced state process $Z_1,\dots,Z_n$ is i.i.d. with marginal $P_Z$.
\end{enumerate}

Then any rate
\[
R < I(Z;X)
\]
is achievable with vanishing block error probability. Consequently, any rate
\[
R < C(\varepsilon)
:=
\sup_{\bigl(P_{X\mid Z=z}\bigr)_{z\in\mathcal{Z}}\in \mathcal{F}(\varepsilon)}
I(Z;X)
\]
is achievable.
\end{theorem}

\begin{proof}
Under the stated assumptions, the induced channel from $Z^n=(Z_1,\dots,Z_n)$ to $X^n=(X_1,\dots,X_n)$ is a discrete memoryless channel with transition law
\[
P_{X^n\mid Z^n}(x^n\mid z^n)=\prod_{t=1}^n P_{X\mid Z}(x_t\mid z_t).
\]
Fix any channel family in $\mathcal{F}(\varepsilon)$ and let the input distribution of $Z$ be the induced marginal $P_Z$. Let $R<I(Z;X)$.

Next, construct a random codebook~\citep{bok:MW} by drawing
\[
z^n(1),\dots,z^n(2^{nR})
\]
independently, each with i.i.d. components distributed according to $P_Z$. To transmit message $m$, send the corresponding codeword $z^n(m)$ through the channel, producing $x^n$ according to
\[
\prod_{t=1}^n P_{X\mid Z}(x_t\mid z_t(m)).
\]
At the decoder, use jointly typical decoding: declare $\widehat m$ to be the unique index such that $(z^n(\widehat m),x^n)$ is jointly typical.

By the standard random-coding theorem for discrete memoryless channels~\citep{bok:MW}, the average probability of error over the random code ensemble tends to zero as $n\to\infty$ whenever
\[
R<I(Z;X).
\]
Hence there exists at least one deterministic codebook with block error probability tending to zero at every rate below $I(Z;X)$.

Since the chosen channel family belongs to $\mathcal{F}(\varepsilon)$, the resulting scheme satisfies the per-key TV constraint. Finally, because the argument holds for every feasible channel family, taking the supremum over all
\[
\bigl(P_{X\mid Z=z}\bigr)_{z\in\mathcal{Z}}\in\mathcal{F}(\varepsilon)
\]
shows that every rate
\[
R<C(\varepsilon)
\]
is achievable.
\end{proof}

\begin{remark}[How to instantiate the constrained channel]
A practical instantiation is to solve, at each time $t$,
\[
\max_{\bigl(P_{X_t\mid Z_t=z}\bigr)_{z\in\mathcal{Z}}}
I(Z_t;X_t)
\quad
\text{subject to}
\quad
d_{\mathrm{TV}}(P_{X_t\mid Z_t=z},Q_t)\le \varepsilon
\ \ \forall z\in\mathcal{Z}.
\]
\end{remark}

\subsection{Examples for Token-Level}

\paragraph{Concrete instantiations of Theorem~\ref{thm:correct-converse}.}
Theorem~\ref{thm:correct-converse} states that under the uniform per-key $\varepsilon$-TV constraint, any sequence of watermarking schemes with vanishing decoding error must satisfy
\[
R_n
\le
\frac1n \sum_{t=1}^n \widetilde C_t(\varepsilon)
+
\frac{1+nR_n P_e^{(n)}}{n},
\]
where
\[
\widetilde C_t(\varepsilon)
:=
\sup_{\substack{(P_{X_t|Z_t=z})_{z\in\mathcal Z}\in \mathcal F_t(\varepsilon)\\ P_{Z_t}\in \Delta(\mathcal Z)}}
I(Z_t;X_t).
\]
Thus, to instantiate the theorem concretely, one fixes a base distribution $Q_t$, chooses state-conditioned channels
\[
\bigl(P_{X_t|Z_t=z}\bigr)_{z\in\mathcal Z}
\]
satisfying
\[
d_{\mathrm{TV}}(P_{X_t|Z_t=z},Q_t)\le \varepsilon
\qquad \forall z\in\mathcal Z,
\]
and then computes the resulting mutual information $I(Z_t;X_t)$. This gives an explicit feasible single-letter quantity that can be compared against the converse bound in Theorem~\ref{thm:correct-converse}.

\paragraph{Example 1: binary symmetric perturbation.}
Let
\[
\mathcal X=\{0,1\},
\qquad
Q(0)=Q(1)=\frac12,
\]
and let the state alphabet be
\[
\mathcal Z=\{0,1\},
\qquad
P_Z(0)=P_Z(1)=\frac12.
\]
Fix $\varepsilon\in[0,\tfrac12]$ and define the conditional laws
\[
P_{X|Z=0}(0)=\frac12+\varepsilon,
\qquad
P_{X|Z=0}(1)=\frac12-\varepsilon,
\]
and
\[
P_{X|Z=1}(0)=\frac12-\varepsilon,
\qquad
P_{X|Z=1}(1)=\frac12+\varepsilon.
\]
We verify feasibility:
\[
d_{\mathrm{TV}}(P_{X|Z=0},Q)
=
\frac12\Bigl(
\bigl|\tfrac12+\varepsilon-\tfrac12\bigr|
+
\bigl|\tfrac12-\varepsilon-\tfrac12\bigr|
\Bigr)
=
\frac12(\varepsilon+\varepsilon)
=
\varepsilon,
\]
and similarly
\[
d_{\mathrm{TV}}(P_{X|Z=1},Q)=\varepsilon.
\]
Hence this channel family is feasible under the uniform per-key $\varepsilon$-TV constraint.

Now compute the induced marginal on $X$:
\[
P_X(0)
=
\frac12\Bigl(\frac12+\varepsilon\Bigr)
+
\frac12\Bigl(\frac12-\varepsilon\Bigr)
=
\frac12,
\]
and similarly
\[
P_X(1)=\frac12.
\]
Thus $X\sim \mathrm{Bernoulli}(\tfrac12)$ and
\[
H(X)=1.
\]
Moreover,
\[
H(X|Z=0)=h_2\Bigl(\frac12+\varepsilon\Bigr),
\qquad
H(X|Z=1)=h_2\Bigl(\frac12+\varepsilon\Bigr),
\]
so
\[
H(X|Z)=h_2\Bigl(\frac12+\varepsilon\Bigr).
\]
Therefore
\[
I(Z;X)
=
H(X)-H(X|Z)
=
1-h_2\Bigl(\frac12+\varepsilon\Bigr).
\]
Consequently,
\[
\widetilde C(\varepsilon)\ge 1-h_2\Bigl(\frac12+\varepsilon\Bigr).
\]
This gives a completely explicit feasible instantiation of the single-letter quantity appearing in Theorem~\ref{thm:correct-converse}.

\paragraph{Numerical specialization of Example 1.}
If $\varepsilon=0.1$, then
\[
I(Z;X)=1-h_2(0.6)\approx 0.0290 \text{ bits/token}.
\]
If $\varepsilon=0.2$, then
\[
I(Z;X)=1-h_2(0.7)\approx 0.1187 \text{ bits/token}.
\]
Thus, even in this simplest binary example, the theorem predicts that as the allowed per-key distortion budget increases, the admissible embedding rate also increases.

\paragraph{Small-$\varepsilon$ expansion for Example 1.}
Using the Taylor expansion of binary entropy around $\tfrac12$,
\[
h_2\Bigl(\frac12+\varepsilon\Bigr)
=
1-\frac{2}{\ln 2}\varepsilon^2+o(\varepsilon^2),
\]
we obtain
\[
I(Z;X)
=
1-h_2\Bigl(\frac12+\varepsilon\Bigr)
=
\frac{2}{\ln 2}\varepsilon^2+o(\varepsilon^2).
\]
Hence in this binary symmetric perturbation example, the mutual information scales quadratically in $\varepsilon$ for small $\varepsilon$.

\paragraph{Example 2: exact distortion-freeness forces zero information.}
Now set $\varepsilon=0$. The per-key TV constraint becomes
\[
d_{\mathrm{TV}}(P_{X|Z=z},Q)=0
\qquad \forall z\in\mathcal Z.
\]
Since total variation distance vanishes if and only if the distributions are identical, we obtain
\[
P_{X|Z=z}=Q
\qquad \forall z\in\mathcal Z.
\]
Therefore the conditional law of $X$ does not depend on $Z$, so $X$ and $Z$ are independent:
\[
P_{X,Z}(x,z)=P_Z(z)Q(x)=P_Z(z)P_X(x).
\]
Hence
\[
I(Z;X)=0.
\]
Thus
\[
\widetilde C_t(0)=0
\qquad \forall t,
\]
and Theorem~\ref{thm:correct-converse} implies that any asymptotically reliable watermarking scheme under exact per-key distortion-freeness must satisfy
\[
\limsup_{n\to\infty} R_n \le 0.
\]
This concretely illustrates the paper's statement that exact per-key distortion-freeness forces zero embedding capacity.

\paragraph{Example 3: ternary perturbation.}
Let
\[
\mathcal X=\{1,2,3\},
\qquad
Q(1)=Q(2)=Q(3)=\frac13,
\]
and again let
\[
\mathcal Z=\{0,1\},
\qquad
P_Z(0)=P_Z(1)=\frac12.
\]
Fix $\varepsilon\in[0,\tfrac13]$ and define
\[
P_{X|Z=0}
=
\left(\frac13+\varepsilon,\frac13,\frac13-\varepsilon\right),
\]
\[
P_{X|Z=1}
=
\left(\frac13-\varepsilon,\frac13,\frac13+\varepsilon\right).
\]
We again verify the TV constraint:
\[
d_{\mathrm{TV}}(P_{X|Z=0},Q)
=
\frac12\left(
\left|\frac13+\varepsilon-\frac13\right|
+
\left|\frac13-\frac13\right|
+
\left|\frac13-\varepsilon-\frac13\right|
\right)
=
\frac12(\varepsilon+0+\varepsilon)
=
\varepsilon,
\]
and similarly
\[
d_{\mathrm{TV}}(P_{X|Z=1},Q)=\varepsilon.
\]
Thus the channel family is feasible.

Now compute the output marginal:
\[
P_X
=
\frac12 P_{X|Z=0}+\frac12 P_{X|Z=1}
=
\left(\frac13,\frac13,\frac13\right)
=
Q.
\]
Hence
\[
H(X)=\log_2 3.
\]
Also,
\[
H(X|Z=0)=H(X|Z=1)
=
H\left(\frac13+\varepsilon,\frac13,\frac13-\varepsilon\right),
\]
so
\[
H(X|Z)
=
H\left(\frac13+\varepsilon,\frac13,\frac13-\varepsilon\right).
\]
Therefore
\[
I(Z;X)
=
\log_2 3
-
H\left(\frac13+\varepsilon,\frac13,\frac13-\varepsilon\right).
\]
This yields another explicit feasible instantiation of the constrained mutual information appearing in Theorem~\ref{thm:correct-converse}.

\paragraph{Example 4: direct \simplemark~instantiation at one time step.}
Fix a time \(t\), and suppose the base next-token distribution is
\[
Q_t = (0.5,0.3,0.2)
\]
over a three-token vocabulary \(\mathcal X=\{x_1,x_2,x_3\}\). Suppose also that the relevant
state alphabet at time \(t\) is binary:
\[
\mathcal Z_t=\{z_0,z_1\},\qquad P_{Z_t}(z_0)=P_{Z_t}(z_1)=\frac12.
\]
For some
\[
\delta\in [0,\min\{\varepsilon,0.2\}],
\]
define
\[
P_{X_t|Z_t=z_0}=(0.5+\delta,0.3,0.2-\delta),
\]
and
\[
P_{X_t|Z_t=z_1}=(0.5-\delta,0.3,0.2+\delta).
\]
Then
\[
d_{\mathrm{TV}}(P_{X_t|Z_t=z_0},Q_t)
=
\frac12\left(|\delta|+0+|-\delta|\right)
=
\delta
\le \varepsilon,
\]
and similarly
\[
d_{\mathrm{TV}}(P_{X_t|Z_t=z_1},Q_t)\le \varepsilon.
\]
Hence this is a feasible per-key channel family.

The output marginal is
\[
P_{X_t}
=
\frac12 P_{X_t|Z_t=z_0}
+
\frac12 P_{X_t|Z_t=z_1}
=
(0.5,0.3,0.2)
=
Q_t.
\]
Therefore
\[
H(X_t)=H(0.5,0.3,0.2).
\]
However, the two conditional entropies are not generally equal. Indeed,
\[
H(X_t|Z_t)
=
\frac12 H(0.5+\delta,0.3,0.2-\delta)
+
\frac12 H(0.5-\delta,0.3,0.2+\delta).
\]
Consequently, the correct mutual information is
\[
\boxed{
I(Z_t;X_t)
=
H(0.5,0.3,0.2)
-
\frac12 H(0.5+\delta,0.3,0.2-\delta)
-
\frac12 H(0.5-\delta,0.3,0.2+\delta).
}
\]
Thus, for this concrete choice of feasible perturbation, Theorem~\ref{thm:correct-converse} yields the explicit converse
\[
R_n
\le
\frac1n\sum_{t=1}^n \widetilde C_t(\varepsilon)
+
\frac{1+nR_n P_e^{(n)}}{n},
\]
with
\[
\widetilde C_t(\varepsilon)
\ge
H(0.5,0.3,0.2)
-
\frac12 H(0.5+\delta,0.3,0.2-\delta)
-
\frac12 H(0.5-\delta,0.3,0.2+\delta).
\]

\paragraph{Example 5: using the actual \simplemark~state variable.}
Recall that the state at time $t$ is
\[
Z_t = (C_M(t)+V_t) \bmod q,
\]
where $C_M(t)\in\mathbb{F}_q$ is the code symbol at time $t$ and
$V_t\sim\mathrm{Unif}(\mathbb{F}_q)$ is the secret key value. To obtain
a concrete toy instantiation, take $q=2$, so
$Z_t=(C_M(t)+V_t)\bmod 2$ takes values in $\mathcal{Z}=\{0,1\}$, the
two vertices of the $1$-simplex; write $z_0:=0$ and $z_1:=1$.
If the base distribution is
\[
Q_t=(0.5,0.3,0.2),
\]
we may choose the feasible channel family
\[
P_{X_t|Z_t=z_0}=(0.5+\delta,0.3,0.2-\delta),
\qquad
P_{X_t|Z_t=z_1}=(0.5-\delta,0.3,0.2+\delta),
\]
where $\delta\le \varepsilon$. Then
$d_{\mathrm{TV}}(P_{X_t|Z_t=z_i},Q_t)\le \varepsilon$ for
$i\in\{0,1\}$, so this is a valid per-key instantiation in the
\simplemark~notation.

If \(P_{Z_t}(z_0)=P_{Z_t}(z_1)=1/2\), then the output marginal is again
\[
P_{X_t}
=
\frac12 P_{X_t|Z_t=z_0}
+
\frac12 P_{X_t|Z_t=z_1}
=
Q_t.
\]
Therefore
\[
H(X_t)=H(0.5,0.3,0.2),
\]
while
\[
H(X_t|Z_t)
=
\frac12 H(0.5+\delta,0.3,0.2-\delta)
+
\frac12 H(0.5-\delta,0.3,0.2+\delta).
\]
Hence the correct per-time mutual information is
\[
\boxed{
I(Z_t;X_t)
=
H(0.5,0.3,0.2)
-
\frac12 H(0.5+\delta,0.3,0.2-\delta)
-
\frac12 H(0.5-\delta,0.3,0.2+\delta).
}
\]
Thus Theorem~\ref{thm:correct-converse}  says that any asymptotically reliable scheme built from such per-time perturbations
must satisfy
\[
\limsup_{n\to\infty} R_n
\le
\limsup_{n\to\infty}
\frac1n\sum_{t=1}^n \widetilde C_t(\varepsilon),
\]
and the displayed quantity above gives an explicit feasible single-letter mutual information value
that lower bounds \(\widetilde C_t(\varepsilon)\).

\paragraph{Interpretation.}
These examples show concretely how Theorem~\ref{thm:correct-converse} should be used in practice. One fixes a base next-token distribution $Q_t$, constructs feasible state-conditioned perturbations lying inside the total variation ball of radius $\varepsilon$ around $Q_t$, computes the induced single-letter mutual information $I(Z_t;X_t)$, and then interprets Theorem~\ref{thm:correct-converse} as saying that no asymptotically reliable watermarking scheme can have rate larger than the time-average of the best such single-letter quantities. In particular, as $\varepsilon\to 0$, every feasible state-conditioned law collapses back to $Q_t$, and the mutual information necessarily goes to zero.

\paragraph{Instantiating $R_n$ and $P_e^{(n)}$ in Theorem~\ref{thm:correct-converse}.}
Recall that in Theorem~\ref{thm:correct-converse}, if the message set has size
\[
|\mathcal M|=2^k,
\]
and the watermark spans $n$ time steps, then the embedding rate is
\[
R_n=\frac{k}{n}
\]
bits per token. Moreover, for a decoder producing an estimate $\widehat M$ from the observation $(X_{1:n},V_{1:n})$, the block error probability is
\[
P_e^{(n)}:=\Pr[\widehat M\neq M].
\]
Thus, to instantiate the theorem concretely, one must specify a finite blocklength $n$, a number of embedded bits $k$, and a decoding rule. 

\paragraph{Example 1: binary toy model with $k=1$ and $n=1$.}
Consider the binary construction
\[
\mathcal X=\{0,1\},
\qquad
\mathcal Z=\{0,1\},
\qquad
P_Z(0)=P_Z(1)=\frac12,
\]
with conditional laws
\[
P_{X|Z=0}(0)=\frac12+\varepsilon,
\qquad
P_{X|Z=0}(1)=\frac12-\varepsilon,
\]
\[
P_{X|Z=1}(0)=\frac12-\varepsilon,
\qquad
P_{X|Z=1}(1)=\frac12+\varepsilon.
\]
Suppose we embed one bit, so that
\[
M\in\{0,1\},
\qquad k=1,
\]
and we use a single time step:
\[
n=1.
\]
Then
\[
R_1=\frac{k}{n}=1
\]
bit per token.

Assume that the state is identified with the message, i.e.\ $Z=M$, and let the decoder be the maximum-likelihood decoder
\[
\widehat M(x)=
\begin{cases}
0,& x=0,\\
1,& x=1.
\end{cases}
\]
Then the block error probability is
\[
P_e^{(1)}=\Pr[\widehat M\neq M].
\]
Since the decoder makes an error exactly when the output flips relative to the more likely symbol under the true message,
\[
P_e^{(1)}
=
\frac12\,\Pr[\widehat M\neq 0\mid M=0]
+
\frac12\,\Pr[\widehat M\neq 1\mid M=1].
\]
Now
\[
\Pr[\widehat M\neq 0\mid M=0]
=
\Pr[X=1\mid Z=0]
=
\frac12-\varepsilon,
\]
and similarly
\[
\Pr[\widehat M\neq 1\mid M=1]
=
\Pr[X=0\mid Z=1]
=
\frac12-\varepsilon.
\]
Therefore
\[
P_e^{(1)}=\frac12-\varepsilon.
\]
So in this one-shot example,
\[
R_1=1,
\qquad
P_e^{(1)}=\frac12-\varepsilon.
\]
This illustrates that at blocklength $1$, the decoder error remains large unless $\varepsilon$ is quite large.

\paragraph{Example 2: repetition over $n$ time steps with $k=1$.}
Now use the same binary channel at every time step, but repeat the same embedded bit across $n$ independent uses. Thus
\[
M\in\{0,1\},
\qquad
k=1,
\qquad
R_n=\frac1n.
\]
Conditioned on $M=0$, the outputs $X_1,\dots,X_n$ are i.i.d.\ with
\[
\Pr[X_t=0\mid M=0]=\frac12+\varepsilon,
\qquad
\Pr[X_t=1\mid M=0]=\frac12-\varepsilon.
\]
Conditioned on $M=1$, the outputs are i.i.d.\ with the roles of $0$ and $1$ reversed.

Let the decoder use majority vote:
\[
\widehat M=
\begin{cases}
0,& \sum_{t=1}^n \mathbf 1\{X_t=1\}<\frac n2,\\
1,& \sum_{t=1}^n \mathbf 1\{X_t=1\}>\frac n2.
\end{cases}
\]
(Without loss of generality, assume that $n$ is odd.)

Then the block error probability is
\[
P_e^{(n)}=\Pr[\widehat M\neq M].
\]
By symmetry,
\[
P_e^{(n)}
=
\Pr\!\left[\sum_{t=1}^n \mathbf 1\{X_t=1\}\ge \frac n2 \,\middle|\, M=0\right].
\]
Since under $M=0$,
\[
\sum_{t=1}^n \mathbf 1\{X_t=1\}
\sim \mathrm{Bin}\!\left(n,\frac12-\varepsilon\right),
\]
we obtain
\[
P_e^{(n)}
=
\sum_{j=\lceil n/2\rceil}^{n}
\binom{n}{j}
\left(\frac12-\varepsilon\right)^j
\left(\frac12+\varepsilon\right)^{n-j}.
\]
Thus in this repetition construction,
\[
R_n=\frac1n,
\]
while
\[
P_e^{(n)}
=
\sum_{j=\lceil n/2\rceil}^{n}
\binom{n}{j}
\left(\frac12-\varepsilon\right)^j
\left(\frac12+\varepsilon\right)^{n-j}.
\]
As $n\to\infty$, this error probability tends to zero for every fixed $\varepsilon>0$, whereas the rate $R_n=1/n$ also tends to zero. This is consistent with Theorem~\ref{thm:correct-converse}.

\paragraph{Example 3: four messages over two time steps.}
Suppose now that we wish to embed
\[
k=2
\]
bits using
\[
n=2
\]
time steps. Then the message set has size
\[
|\mathcal M|=2^2=4,
\]
for example
\[
\mathcal M=\{00,01,10,11\}.
\]
The embedding rate is therefore
\[
R_2=\frac{k}{n}=\frac22=1
\]
bit per token.

To make this concrete, suppose at each time $t$ the encoder uses the binary perturbation channel from Example~1, and the two message bits are sent independently across the two time steps:
\[
M=(M_1,M_2),
\qquad
X_t \sim P_{X|Z=M_t}.
\]
Let the decoder decode each coordinate separately via
\[
\widehat M_t=
\begin{cases}
0,& X_t=0,\\
1,& X_t=1.
\end{cases}
\]
Then the overall decoded message is
\[
\widehat M=(\widehat M_1,\widehat M_2).
\]
Since an error occurs if at least one of the two coordinates is decoded incorrectly, and each coordinate has error probability
\[
\frac12-\varepsilon,
\]
the block error probability is
\[
P_e^{(2)}
=
1-\Pr[\widehat M_1=M_1 \text{ and } \widehat M_2=M_2]
=
1-\left(\frac12+\varepsilon\right)^2.
\]
Equivalently,
\[
P_e^{(2)}
=
1-\left(1-\left(\frac12-\varepsilon\right)\right)^2.
\]
So in this example,
\[
R_2=1,
\qquad
P_e^{(2)}=1-\left(\frac12+\varepsilon\right)^2.
\]

\paragraph{Example 4: general finite-blocklength instantiation.}
More generally, if one chooses:
\begin{itemize}
    \item a blocklength $n$,
    \item a message set $\mathcal M$ with $|\mathcal M|=2^k$,
    \item a feasible family of per-time channels
    \[
    \bigl(P_{X_t|Z_t=z}\bigr)_{z\in\mathcal Z},
    \]
    \item and a decoder
    \[
    \widehat M=\widehat M(X_{1:n},V_{1:n}),
    \]
\end{itemize}
then the two quantities in Theorem~\ref{thm:correct-converse} are instantiated as
\[
R_n=\frac{k}{n}
\]
and
\[
P_e^{(n)}=\Pr[\widehat M(X_{1:n},V_{1:n})\neq M].
\]
For instance, if $k=100$ bits are embedded over $n=1000$ tokens, then
\[
R_n=\frac{100}{1000}=0.1
\]
bits per token. If the chosen encoder--decoder pair achieves a block error probability of, say,
\[
P_e^{(n)}=10^{-3},
\]
then the finite-$n$ converse in Theorem~\ref{thm:correct-converse} becomes
\[
0.1
\le
\frac1{1000}\sum_{t=1}^{1000}\widetilde C_t(\varepsilon)
+
\frac{1+1000\cdot 0.1\cdot 10^{-3}}{1000}.
\]
That is,
\[
0.1
\le
\frac1{1000}\sum_{t=1}^{1000}\widetilde C_t(\varepsilon)
+
0.0011.
\]
Equivalently,
\[
\frac1{1000}\sum_{t=1}^{1000}\widetilde C_t(\varepsilon)
\ge 0.0989.
\]
So this gives a concrete numerical reading of Theorem~\ref{thm:correct-converse}: to support an embedding rate of $0.1$ bits per token at block error probability $10^{-3}$, the average constrained single-letter mutual information must be at least approximately $0.0989$.

\paragraph{Interpretation.}
In short, $R_n$ is instantiated simply by choosing how many bits $k$ are embedded across how many time steps $n$, while $P_e^{(n)}$ is instantiated by fixing a decoder and computing its block error probability under the induced joint law of $(M,X_{1:n},V_{1:n})$. Theorem~\ref{thm:correct-converse} then relates these concrete finite-blocklength quantities to the average single-letter constrained capacities $\widetilde C_t(\varepsilon)$.

\subsection{Sequence-Level}

\begin{lemma}[Sequence-level TV bound from uniform per-step TV control]
\label{lem:sequence-tv}
Fix a message $m\in\mathcal{M}$ and a realized side-information sequence $v_{1:n}\in\mathbb{F}_q^n$. Let
\[
P^{(m,v_{1:n})}_{X_{1:n}}
:=
P_{X_{1:n}\mid M=m,V_{1:n}=v_{1:n}}
\]
denote the watermarked sequence law, and let
\[
Q^{(n)}_{X_{1:n}}
:=
\prod_{t=1}^n Q_t
\]
denote the reference product law formed from the base next-token distributions.
Assume that, conditioned on $(M,V_{1:n})=(m,v_{1:n})$, the watermarked sequence law factors as
\[
P^{(m,v_{1:n})}_{X_{1:n}}(x_{1:n})
=
\prod_{t=1}^n P_{X_t\mid M=m,V_t=v_t}(x_t),
\]
and that for every $t\in[1:n]$,
\[
d_{\mathrm{TV}}\!\left(P_{X_t\mid M=m,V_t=v_t},\,Q_t\right)\le \varepsilon_t.
\]
Then
\[
d_{\mathrm{TV}}\!\left(P^{(m,v_{1:n})}_{X_{1:n}},\,Q^{(n)}_{X_{1:n}}\right)
\le
1-\prod_{t=1}^n (1-\varepsilon_t).
\]
In particular, if $\varepsilon_t\le \varepsilon$ for all $t$, then
\[
d_{\mathrm{TV}}\!\left(P^{(m,v_{1:n})}_{X_{1:n}},\,Q^{(n)}_{X_{1:n}}\right)
\le
1-(1-\varepsilon)^n
\le
n\varepsilon.
\]
\end{lemma}

\begin{proof}
For each $t$, write
\[
P_t := P_{X_t\mid M=m,V_t=v_t},
\qquad
Q_t := Q_t.
\]
By assumption,
\[
P^{(m,v_{1:n})}_{X_{1:n}} = \bigotimes_{t=1}^n P_t,
\qquad
Q^{(n)}_{X_{1:n}} = \bigotimes_{t=1}^n Q_t.
\]

For each $t$, by the coupling characterization of total variation distance, there exists a coupling
\[
\pi_t \in \Pi(P_t,Q_t)
\]
of random variables $(Y_t,Z_t)$ such that
\[
\Pr[Y_t\neq Z_t] = d_{\mathrm{TV}}(P_t,Q_t)\le \varepsilon_t.
\]
Take these couplings independently across $t$, producing a joint coupling of
\[
Y_{1:n}=(Y_1,\dots,Y_n)\sim \bigotimes_{t=1}^n P_t
\]
and
\[
Z_{1:n}=(Z_1,\dots,Z_n)\sim \bigotimes_{t=1}^n Q_t.
\]
Under this product coupling,
\[
\Pr[Y_{1:n}=Z_{1:n}]
=
\prod_{t=1}^n \Pr[Y_t=Z_t]
=
\prod_{t=1}^n \bigl(1-\Pr[Y_t\neq Z_t]\bigr)
\ge
\prod_{t=1}^n (1-\varepsilon_t).
\]
Therefore,
\[
\Pr[Y_{1:n}\neq Z_{1:n}]
=
1-\Pr[Y_{1:n}=Z_{1:n}]
\le
1-\prod_{t=1}^n (1-\varepsilon_t).
\]
Applying again the coupling characterization of total variation distance gives
\[
d_{\mathrm{TV}}\!\left(\bigotimes_{t=1}^n P_t,\ \bigotimes_{t=1}^n Q_t\right)
\le
\Pr[Y_{1:n}\neq Z_{1:n}]
\le
1-\prod_{t=1}^n (1-\varepsilon_t).
\]
This proves the first claim.

For the second claim, if $\varepsilon_t\le \varepsilon$ for all $t$, then
\[
1-\prod_{t=1}^n (1-\varepsilon_t)
\le
1-(1-\varepsilon)^n.
\]
Finally, Bernoulli's inequality implies
\[
(1-\varepsilon)^n \ge 1-n\varepsilon
\qquad \text{for } \varepsilon\in[0,1],
\]
and hence
\[
1-(1-\varepsilon)^n \le n\varepsilon.
\]
This proves the uniform bound.
\end{proof}

\begin{corollary}[Uniform sequence-level TV guarantee]
\label{cor:uniform-sequence-tv}
Assume the hypotheses of Lemma~\ref{lem:sequence-tv}, and suppose moreover that the scheme is uniformly per-key $\varepsilon$-TV distortion-free, i.e.
\[
d_{\mathrm{TV}}
\bigl(
P_{X_t\mid M=m,V_t=v},
Q_t
\bigr)\le \varepsilon
\qquad
\forall t,\ \forall m,\ \forall v.
\]
Then for every message $m$ and every realized side-information sequence $v_{1:n}$,
\[
d_{\mathrm{TV}}
\!\left(
P_{X_{1:n}\mid M=m,V_{1:n}=v_{1:n}},
\ \prod_{t=1}^n Q_t
\right)
\le
1-(1-\varepsilon)^n
\le
n\varepsilon.
\]
\end{corollary}

\begin{proof}
Apply Lemma~\ref{lem:sequence-tv} with $\varepsilon_t=\varepsilon$ for all $t$.
\end{proof}
\section{Distortion-Freeness: Per-Key vs. On-Average}

In \arcmark~\citep{gilani2026arcmarkmultibitllmwatermark}, they consider
distortion-freeness, averaged over the secret key. For many applications,
this is not
an acceptable level of security since LLM watermarking of images/text 
is done for a particular secret key.

\begin{proposition}
The perfect distortion-free condition averaged over the secret key,
$\mathbb{E}_{S}\!\left[P_{X\mid S,M=m}\right]=Q_X
\quad \forall m,$
does not imply per-key perfect distortion-freeness, namely
$P_{X\mid S=s,M=m}=Q_X \quad \forall s,m.$
\end{proposition}

\begin{proof}
Consider the binary alphabet $\mathcal X=\{0,1\}$ and let
$Q_X(0)=Q_X(1)=\tfrac12$.
Let the secret key $S$ be uniform on $\{0, 1\}$:
$\Pr[S=0]=\Pr[S=1]=\tfrac12.$
For every message $m$, define
$P_{X\mid S=0,M=m}(0)=1,P_{X\mid S=0,M=m}(1)=0,$
and
$P_{X\mid S=1,M=m}(0)=0,P_{X\mid S=1,M=m}(1)=1.$
Then
$\mathbb{E}_{S}\!\left[P_{X\mid S,M=m}(0)\right]
=
\tfrac12\cdot 1+\tfrac12\cdot 0
=
\tfrac12
=
Q_X(0),$
and similarly
$\mathbb{E}_{S}\!\left[P_{X\mid S,M=m}(1)\right]
=
\tfrac12\cdot 0+\tfrac12\cdot 1
=
\tfrac12
=
Q_X(1).$
Hence
$\mathbb{E}_{S}\!\left[P_{X\mid S,M=m}\right]=Q_X
\quad \forall m.$
However,
$P_{X\mid S=0,M=m}=\delta_0 \neq Q_X,
P_{X\mid S=1,M=m}=\delta_1 \neq Q_X.$
Thus per-key perfect distortion-freeness fails.
\end{proof}

\subsection{Separation in Detection}

We show that schemes that 
provide average-key distortion-free guarantees 
could be strictly less covert against a detector/adversary than the class of uniformly per-key \(\varepsilon\)-TV distortion-free schemes. 
In particular, for every \(\varepsilon < 1\), there exist average-key distortion-free schemes that are strictly easier to detect than any scheme obeying the uniform per-key \(\varepsilon\)-TV constraint.

In the proposition below, the
detector explicitly receives the secret key. In a real-world
deployment, this is (likely) not the case. But we use the
proposition to show the separation in detection schemes that
satisfy average-key distortion freeness versus per-key
distortion freeness.

\begin{proposition}[Average-key distortion-freeness can be highly detectable]
\label{prop:average-key-detectable}
Let \(\mathcal X=\mathcal V=\{1,\ldots,N\}\), and let
\(Q(x)=\frac1N, x\in\mathcal X.\)
Let the secret key \(V\sim \mathrm{Unif}(\mathcal V)\). For each realized key value \(v\in\mathcal V\),
define the keyed output distribution
\(P_{X|V=v}=\delta_v.\)
Then the following holds:

\begin{enumerate}
    \item The scheme is perfectly distortion-free on average over the key:
    \(\mathbb E_V[P_{X|V}] = Q.\)

    \item The scheme is far from per-key distortion-free. For every \(v\in\mathcal V\),
    \(d_{\mathrm{TV}}(P_{X|V=v},Q) = d_{\mathrm{TV}}(\delta_v,Q) =
    1-\frac1N. \)

    \item If the detector observes both \(X\) and the realized key \(V\), then the null law
    \( H_0:\quad V\sim \mathrm{Unif}(\mathcal V),\quad X\sim Q,\quad X\perp V \)
    and the watermarked law
    \(H_1:\quad V\sim \mathrm{Unif}(\mathcal V),\quad X=V\)
    have total variation distance
    \(d_{\mathrm{TV}}(P_{X,V}^{H_1},P_{X,V}^{H_0})
    =
    1-\frac1N.\)
    In particular, the test
    \(\Phi(x,v)=\mathbf 1\{x=v\}\)
    satisfies
    \[
    \mathbb P_{H_1}[\Phi(X,V)=1]=1,
    \quad
    \mathbb P_{H_0}[\Phi(X,V)=1]=\frac1N,
    \]
    and hence has distinguishing advantage \(1-1/N\).
\end{enumerate}
Consequently, perfect average-key distortion-freeness does not prevent strong detectability once the
realized key is used by the detector. Moreover, for every \(\varepsilon<1\), choosing \(N\) large enough
so that
\(1-\frac1N>\varepsilon\)
gives an average-key perfectly distortion-free scheme whose keyed detection advantage exceeds the
per-key TV budget \(\varepsilon\).
\end{proposition}

\begin{proof}
First, for each \(x\in\mathcal X\),
\(\mathbb E_V[P_{X|V}(x)]
=
\frac1N\sum_{v=1}^N \delta_v(x)
=
\frac1N
=
Q(x).\)
Therefore
\(\mathbb E_V[P_{X|V}]=Q,\)
so the scheme is perfectly distortion-free on average over the key.

Next, fix \(v\in\mathcal V\). Since \(P_{X|V=v}=\delta_v\) and \(Q\) is uniform,
\(d_{\mathrm{TV}}(\delta_v,Q)
=
\frac12\sum_{x=1}^N |\delta_v(x)-Q(x)|.
\)
The term \(x=v\) contributes
\(
\left|1-\frac1N\right|=1-\frac1N,
\)
while the remaining \(N-1\) terms contribute
\(
\sum_{x\neq v}\left|0-\frac1N\right|
=
\frac{N-1}{N}
=
1-\frac1N.
\)
Thus
\(
d_{\mathrm{TV}}(\delta_v,Q)
=
\frac12\left(1-\frac1N+1-\frac1N\right)
=
1-\frac1N.
\)

It remains to compute the joint distinguishability: Under \(H_1\),
\[
P_{X,V}^{H_1}(x,v)
=
\begin{cases}
\frac1N, & x=v,\\
0, & x\neq v.
\end{cases}
\]
Under \(H_0\), \(X\) and \(V\) are independent uniforms, so
\(
P_{X,V}^{H_0}(x,v)=\frac1{N^2}
\quad \forall x,v.
\)
Therefore
\[
d_{\mathrm{TV}}(P_{X,V}^{H_1},P_{X,V}^{H_0})
=
\frac12
\sum_{v=1}^N
\sum_{x=1}^N
\left|
P_{X,V}^{H_1}(x,v)-P_{X,V}^{H_0}(x,v)
\right|.
\]
For the \(N\) diagonal pairs \(x=v\), the absolute difference is
\(
\left|\frac1N-\frac1{N^2}\right|
=
\frac{N-1}{N^2}.
\)
For the \(N^2-N\) off-diagonal pairs \(x\neq v\), the absolute difference is
\(
\left|0-\frac1{N^2}\right|
=
\frac1{N^2}.
\)
Hence
\(
d_{\mathrm{TV}}(P_{X,V}^{H_1},P_{X,V}^{H_0})
=
\frac12
\left[
N\cdot \frac{N-1}{N^2}
+
(N^2-N)\cdot \frac1{N^2}
\right].
\)
Simplifying,
\[
d_{\mathrm{TV}}(P_{X,V}^{H_1},P_{X,V}^{H_0})
=
\frac12
\left[
\frac{N-1}{N}
+
\frac{N-1}{N}
\right]
=
1-\frac1N.
\]
Finally, for the test
\(
\Phi(x,v)=\mathbf 1\{x=v\},
\)
we have under \(H_1\) that \(X=V\) always, so
\(
\mathbb P_{H_1}[\Phi(X,V)=1]=1.
\)
Under \(H_0\), \(X\) and \(V\) are independent uniform random variables on \(\{1,\ldots,N\}\), so
\(
\mathbb P_{H_0}[\Phi(X,V)=1]
=
\sum_{v=1}^N \mathbb P[V=v]\mathbb P[X=v]
=
N\cdot \frac1N\cdot \frac1N
=
\frac1N.
\)
Thus the test has distinguishing advantage
\(
1-\frac1N.
\)
This proves the proposition.
\end{proof}

\section{Additional Experiments}
\label{sec:additional-experiments}
We provide additional experimental results. Unless specified in the experiment, the setting is identical to that of Section~\ref{sec:experiments}.

\subsection{Long-form decoding}
We experiment with decoding lower $\varepsilon$-budgets under a 500-1000 token length sequence. \simplemark's decoder is able to recover the messages quite accurately in this regime. We use the Qwen-3.5-9B-Base model. The results are shown in Figure~\ref{fig:LF}.

\begin{figure}
    \centering
    \includegraphics[width=1.0\linewidth]{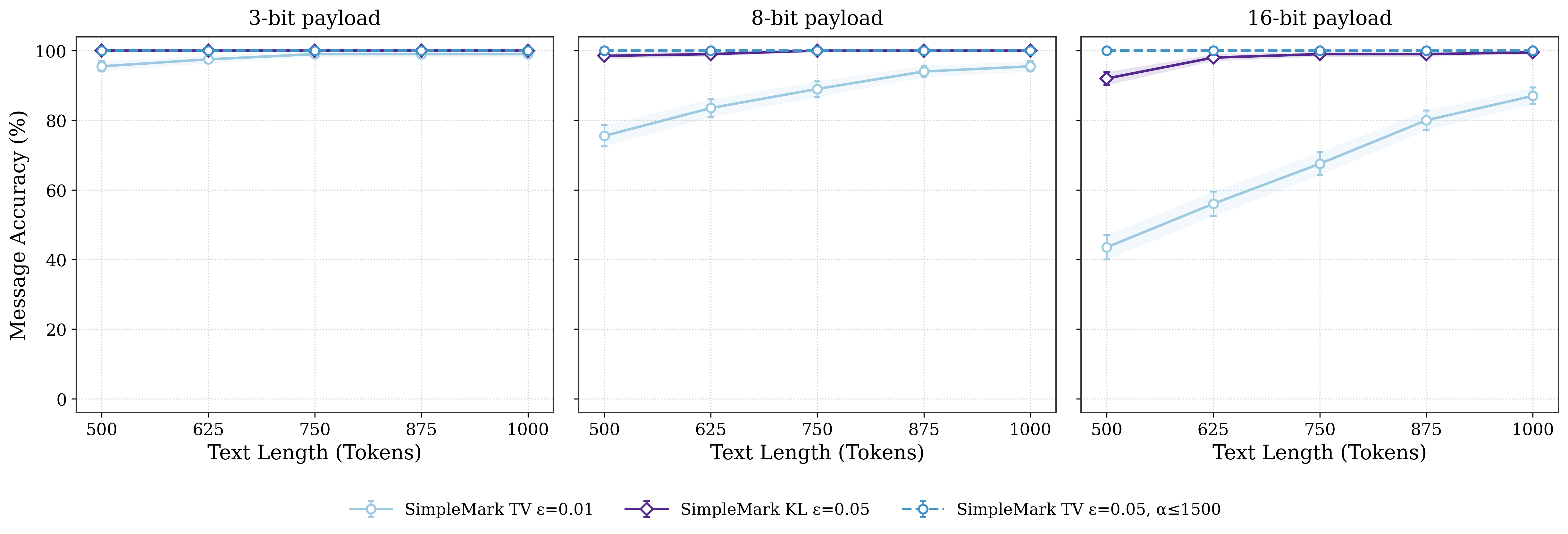}
    \caption{Message recovery for a variable length prefix of a 500-1000 token sequence at various $\varepsilon$ budgets.}
    \label{fig:LF}
\end{figure}

\subsection{Llama3.1-8B-Instruct Reasoning}
We include the Llama3.1-8B-Instruct reasoning results here. All results are consistent with the unwatermarked output in a TV or KL $\varepsilon = 0.05$ decodable regime. AIME was omitted from the table since out of the 2024 and 2025 exams, no model configuration got over 2 problems correct.

\begin{table}[t]
\centering
\small
\setlength{\tabcolsep}{4.5pt}
\begin{tabular}{lccc}
\toprule
Scheme & GSM8K ($n{=}300$) & GPQA-D ($n{=}198$) & LCB ($n{=}300$) \\
\midrule
\multicolumn{4}{l}{\textit{3-bit payload}} \\
SimpleMark TV ($\varepsilon{=}0.05$, $\alpha_{\max}{=}1500$) & 253 & 58 & 62 \\
SimpleMark KL ($\varepsilon{=}0.05$, $\alpha_{\max}{=}1500$) & 249 & 42 & 54 \\
BiMark & 251 & 45 & 58 \\
ArcMark & 249 & 54 & 65 \\
Unwatermarked & 249 & 56 & 60 \\
\midrule
\multicolumn{4}{l}{\textit{8-bit payload}} \\
SimpleMark TV ($\varepsilon{=}0.05$, $\alpha_{\max}{=}1500$) & 252 & 55 & 60 \\
SimpleMark KL ($\varepsilon{=}0.05$, $\alpha_{\max}{=}1500$) & 247 & 61 & 62 \\
BiMark & 248 & 53 & 55 \\
ArcMark & 259 & 56 & 57 \\
Unwatermarked & 249 & 56 & 60 \\
\bottomrule
\end{tabular}
\caption{Reasoning performance on several tasks for Llama3.1-8B-Instruct (temperature 0.6) on GSM8K (300 problems$\times$2 independent generations), GPQA-Diamond (198 problems), and LiveCodeBench (300 problems, graded on public test cases only, so counts are only comparable across schemes). We use identical tasks and datasets for every benchmark. We evaluate SimpleMark at two operating points ($q{=}53$, TV or KL $\varepsilon = 0.05$, cap $\alpha_{\max}{=}1500$), BiMark, ArcMark, and no watermark.}
\label{tab:reasoning_task_accuracy_llama}
\end{table}

\subsection{Robustness to token substitution}
Figure~\ref{fig:robust} depicts the message decoding accuracy for a 300 token sequence when a $\%$ of the tokens are substituted with a random word in the vocabulary. As expected, the message recovery exponentially degrades as the substitution \% increases.

\begin{figure}
    \centering
    \includegraphics[width=1.0\linewidth]{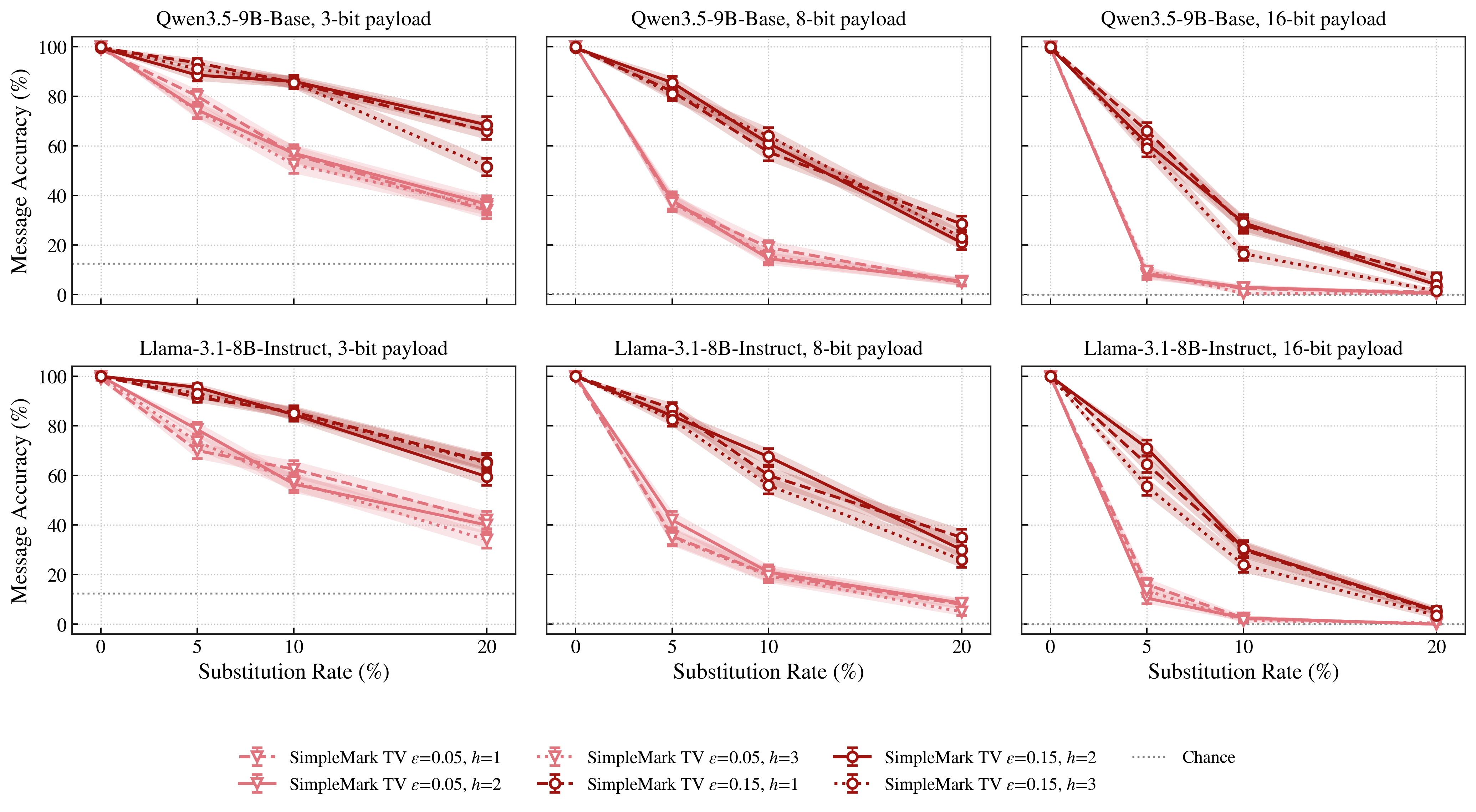}
    \caption{Substitution for a variable percentage of tokens in a 300 token sequence. For $\varepsilon \in \{0.05, 0.15\}$, and context hashing window $h = 2,3$. \simplemark is decently robust to substitution edit attacks.}
    \label{fig:robust}
\end{figure}

\subsection{Zero-Bit Detection}

Figure~\ref{fig:zerobit} reports the true positive rate of zero-bit
detection as a function of text length at fixed false positive rates
of $10^{-2}$, $10^{-6}$, and $10^{-12}$, on Qwen3.5-9B on the top and Llama3.1-8B Instruct on the bottom. Rapid zero-bit detection is not an unambiguous virtue: the keyed
detector's power is fueled by the same distributional deviation an
adversary can exploit, and the baselines detect quickly precisely
because their realized per-token deviations are large (46--65\%
distinguishing advantage, Figure~\ref{fig:TV}), far in excess of what
reliable detection requires. 
\simplemark instead caps the keyed per-token distinguishing advantage by $\varepsilon$.

\begin{figure}
    \centering
    \includegraphics[width=1.0\linewidth]{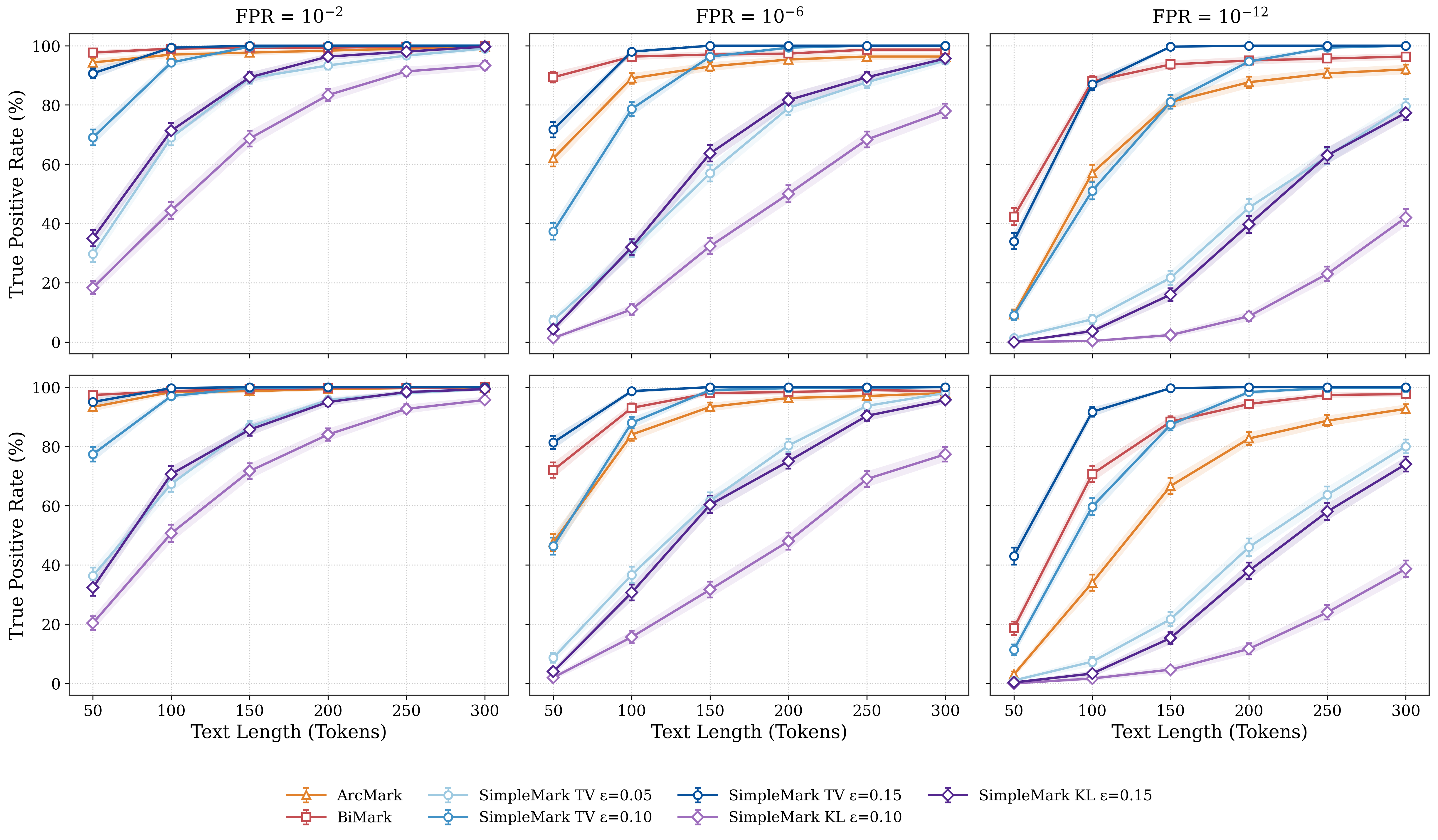}
    \caption{Zero-bit detection for \simplemark, \arcmark, and Bimark at different FPR. At higher distortion budgets we see overall improvement on the baselines for detection rate. For lower distortion budgets we stay closer to the unwatermarked distribution, so naturally we see lower detection rates.}
    \label{fig:zerobit}
\end{figure}

\subsection{Ablation on different field, token sampling, and simplex size $q$}

Throughout the paper we restricted \simplemark's simplex, LDPC code field size, and top-$k$ sampling all to 53. In this section we rerun experiments with a smaller parameter of 29. We include message decoding~\ref{fig:q_ab_deco}, perplexity~\ref{fig:q_ab_perp}, and GSM8K reasoning~\ref{fig:q_ab_gm} baselines with $\alpha_{\max} \in \{200, 500, 1000\}$. All experiments are exactly the same, besides the simplex, LDPC code, and top-$k$ sampling size. We use Qwen-3.5-9B-Base for the reasoning and embed a 16-bit payload.

\begin{figure}
    \centering
    \includegraphics[width=1.0\linewidth]{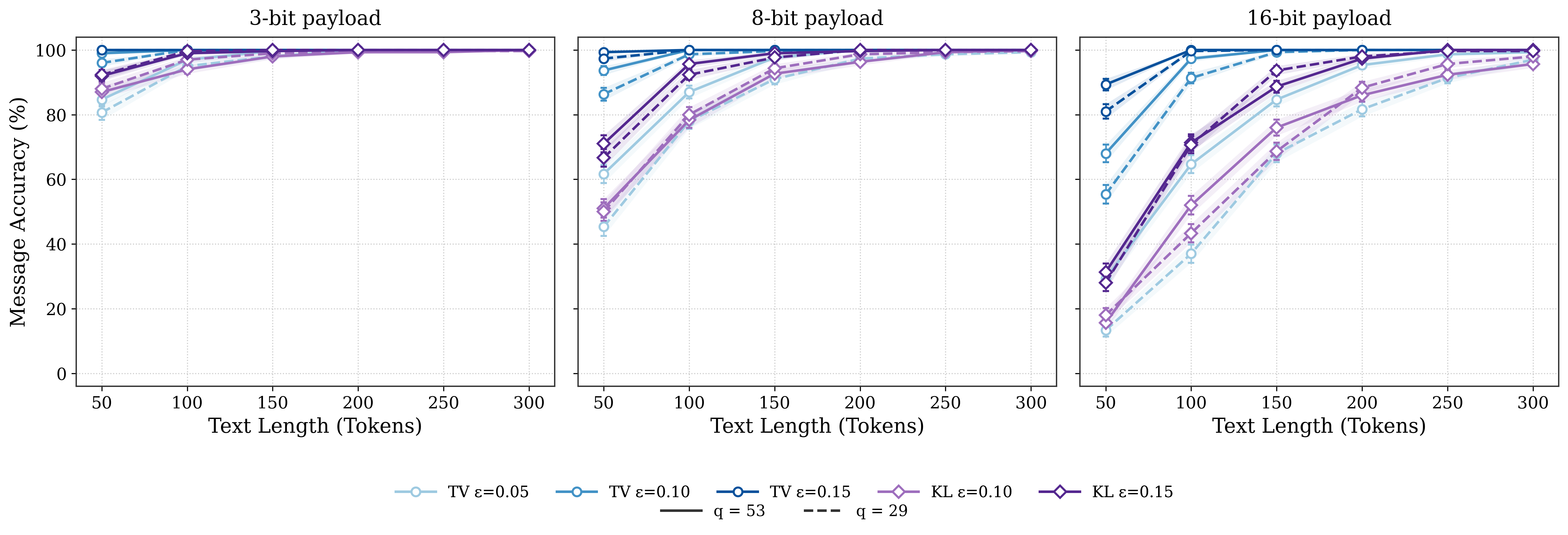}
    \caption{Message recovery for a variable length prefix of a 300 token sequence for $q = \{29, 53\}$ at various $\varepsilon$ budgets.}
    \label{fig:q_ab_deco}
\end{figure}

\begin{figure}
    \centering
    \includegraphics[width=0.5\linewidth]{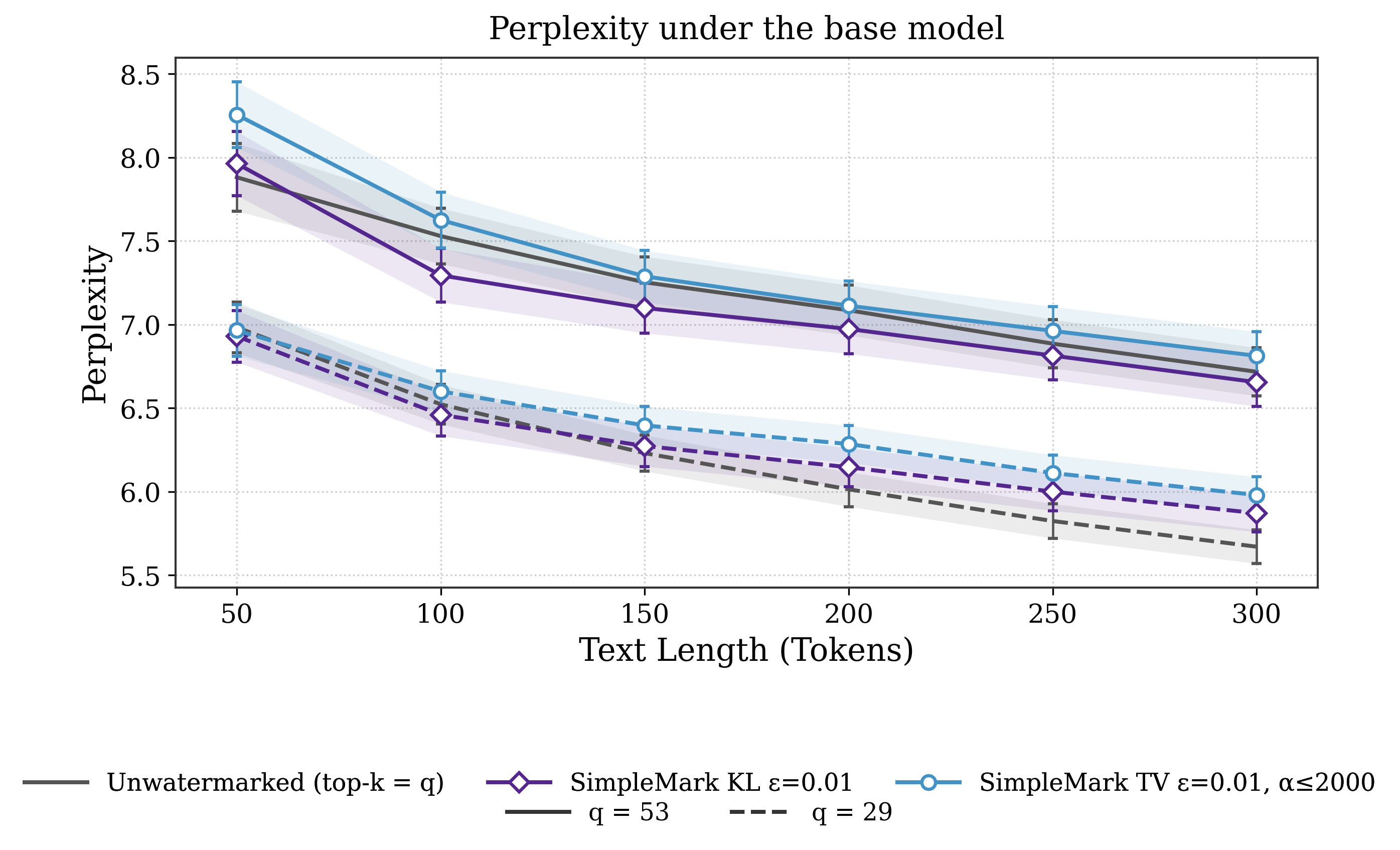}
    \caption{Perplexity with different parameters for $q = \{29, 53\}$. Both remain close to the unwatermarked baselines.}
    \label{fig:q_ab_perp}
\end{figure}

\begin{figure}
    \centering
    \includegraphics[width=0.5\linewidth]{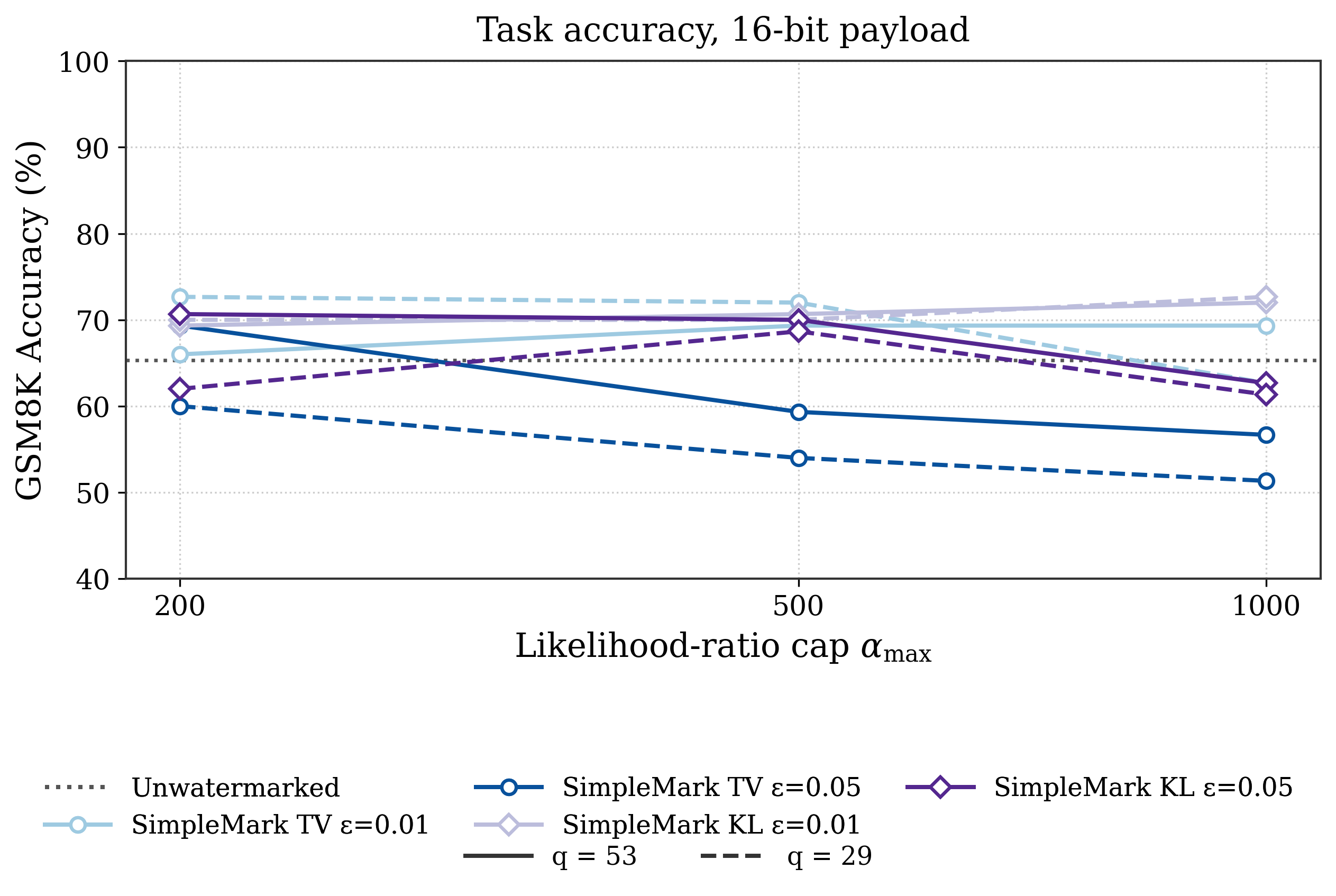}
    \caption{Reasoning for \simplemark under the  different parameter sizes and $\alpha_{\max}$.}
    \label{fig:q_ab_gm}
\end{figure}

\end{document}